\documentclass{IEEEtran}
\usepackage{cite}
\usepackage{amsmath,amssymb,amsfonts}
\usepackage{algorithmic}
\usepackage{graphicx}
\usepackage{textcomp}

\usepackage{amsthm}
\usepackage{mathtools}
\usepackage{color}
\usepackage{xcolor}
\usepackage{subcaption}
\usepackage{textcomp}
\usepackage{array}
\usepackage{lipsum}
\usepackage{multirow}
\allowdisplaybreaks[4]

\newtheorem{definition}{Definition}
\newtheorem{assumption}{Assumption}

\newtheorem{theorem}{Theorem}
\newtheorem{remark}{Remark}
\newtheorem{lemma}{Lemma}

\newtheorem{corollary}{Corollary}

\usepackage{bm}
\usepackage{footmisc}
\usepackage{algorithm}

\DeclareCaptionLabelFormat{custom}{#1(#2)}

\def\BibTeX{{\rm B\kern-.05em{\sc i\kern-.025em b}\kern-.08em
    T\kern-.1667em\lower.7ex\hbox{E}\kern-.125emX}}
\begin{document}
\title{Decentralized Optimization with Amplified Privacy via Efficient Communication}
\author{Wei Huo, Changxin Liu, \IEEEmembership{Member, IEEE}, Kemi Ding, Karl Henrik Johansson, \IEEEmembership{Fellow, IEEE},  Ling Shi, \IEEEmembership{Fellow, IEEE}
\thanks{W. Huo and L. Shi are with the Department of Electronic and Computer Engineering, Hong Kong University of Science and Technology, Clear Water Bay, Kowloon, Hong Kong (email: whuoaa@connect.ust.hk, eesling@ust.hk).}
\thanks{C. Liu is with the Key Laboratory of Smart Manufacturing in Energy Chemical Process, Ministry of Education, East China University of Science and Technology, Shanghai 200237, China (email: changxinl@ecust.edu.cn).}
\thanks{K. H. Johansson is with the Division of Decision and Control Systems, School of Electrical Engineering and Computer Science, KTH Royal Institute of Technology, and also with Digital Futures, SE-10044 Stockholm, Sweden (email: kallej@kth.se).}
\thanks{K. Ding is with the School of System Design and Intelligent Manufacturing, Southern University of Science and Technology, Shenzhen, 518055, China (email: dingkm@sustech.edu.cn).}}

\maketitle

\begin{abstract}
Decentralized optimization is crucial for multi-agent systems, with significant concerns about communication efficiency and privacy.
{This paper explores the role of efficient communication in decentralized stochastic gradient descent algorithms for enhancing privacy preservation. 
We develop a novel algorithm that incorporates two key features: random agent activation and sparsified communication.}
Utilizing differential privacy, we demonstrate that these features reduce noise without sacrificing privacy, thereby amplifying the privacy guarantee and improving accuracy.
Additionally, we analyze the convergence and the privacy-accuracy-communication trade-off of the proposed algorithm.
Finally, we present experimental results to illustrate the effectiveness of our algorithm.
\end{abstract}


\section{Introduction}
\label{sec:introduction}
Decentralized optimization enables multiple agents to collaborate in minimizing a shared objective without a central coordinator.
This approach has diverse applications, including large-scale model training in artificial intelligence~\cite{lian2017can} and spectrum sensing in cognitive networks~\cite{zeng2010distributed}.
Various decentralized optimization methods have been developed, such as decentralized stochastic gradient descent (SGD)~\cite{lian2017can}, gradient tracking~\cite{pu2021distributed} and alternating direction method of multipliers~\cite{shi2014linear}.
Despite the advancement of decentralized optimization techniques, issues such as communication overload and privacy breaches arise from agents broadcasting messages, especially in large-scale networks with mobile devices.
For instance, in GPT-2~\cite{radford2019language}, which comprises approximately 1.5 billion parameters, transmitting these parameter at every iteration would result in prohibitively high communication costs. 
Additionally, storing data locally on agents does not provide sufficient privacy protection, as demonstrated by recent inference attacks in~\cite{zhu2019deep, liu2022membership}.
Therefore, there is a pressing need for integrated algorithms that effectively address communication and privacy challenges.

To address the communication bottleneck in decentralized optimization problems, some studies focus on reducing the number of communication rounds. 
For example, event-triggered methods have been used to balance the trade-off between communication costs and convergence performance~\cite{huo2024distributed}.
Other works employ compression techniques, such as quantization methods to decrease transmitted bit count~\cite{reisizadeh2019exact} and sparsification techniques to reduce message size during transmission~\cite{wang2021error}. 
However, these approaches did not explicitly consider privacy preservation. 

Decentralized optimization with privacy protection increasingly relies on differential privacy (DP) as a leading standard. 
DP, pioneered by Dwork~\cite{dwork2006differential}, is favored for its privacy assurances and efficiency in data analysis tasks.
Huang et al.~\cite{huang2015differentially} introduced a differentially private distributed optimization algorithm by perturbing transmitted data with Laplacian noise. However, this method slowed convergence due to the linearly decaying learning rate.
Ding et al.~\cite{ding2021differentially} enhanced the convergence rate to linear based on a gradient tracking scheme, but at the expense of increased communication load due to the necessity for agents to transmit two variables. 
While stochastic quantization schemes have been employed solely to ensure DP, the DP is established only for each iteration, risking cumulative privacy loss over multiple iterations~\cite{wang2022quantization, huo2024compression}.
Liu et al.~\cite{liu2024distributed} developed a class of differentially private distributed dual averaging algorithms and implemented agent subsampling procedures to bolster DP guarantees. For a comprehensive overview, see~\cite{liu2024survey}. Despite these advancements, the applicability of these methodologies is confined to convex optimizations, whereas many practical tasks entail non-convex optimization problems, such as neural network configurations in deep learning.



Previous research has predominantly focused on either saving communication costs or preserving privacy in decentralized convex optimization~\cite{yi2018dynamic, huo2024distributed, reisizadeh2019exact, wang2021error, huang2015differentially, ding2021differentially, liu2024distributed}, we delve into decentralized non-convex optimization problems under joint communication and DP constraints.
Recent works by Xie et al.~\cite{xie2023compressed, xie2023differentially} integrated compression and DP for decentralized optimization. However, their approaches treating communication reduction and perturbation as independent operations led to cumulative errors, degrading optimization accuracy over finite iterations.
{For a centralized and convex setup, Farokhi~\cite{farokhi2021gradient} investigated the impact of gradient sparsification on DP, and developed differentially private SGD with sparsification that achieves a better privacy-accuracy tradeoff, provided that the privacy budget is tight enough.}
{In the federated setting, some studies used compression during uplink transmission to enhance privacy~\cite{hu2023federated, chen2024privacy}. However, Chen at al.~\cite{chen2024privacy} only considered central DP, which is weaker than agent-level privacy. 
Additionally, communication in decentralized networks is much more complicated than in master-slave settings.}
To address these challenges, 
we aim to investigate the interplay between communication efficiency and DP.
Integrating communication reduction with noise perturbation to enhance privacy and quantifying this enhancement present ongoing challenges.
To overcome these obstacles, we propose a decentralized optimization algorithm based on the framework shown in Fig.~\ref{fig: framework}. 
\begin{figure}[t]  
	\centering
	\includegraphics[width=0.8\linewidth]{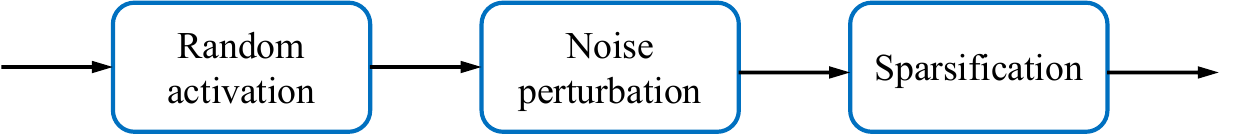}
	\caption{DP amplified by random activation and sparsification.} 
	\label{fig: framework}
\end{figure}
In our approach, agents are randomly activated during each iteration, with only active agents adding noise to their update, followed by message sparsification.
This framework improves communication efficiency by reducing communication rounds and transmitted data size.
Random agent activation ensures that only a randomly selected subset of data samples contributes to gradient calculations, thus improving privacy protection.
Sparsification also reduces gradient sensitivity, leading to decreased privacy loss in each communication round. 
Our method allows for reducing noise intensity without compromising DP, achieving a better balance between privacy, accuracy, and communication efficiency. Theoretical analysis is provided to demonstrate the convergence of our approach and its privacy guarantee.
The main contributions of our work are summarized as follows:
\begin{itemize}
\item[1)] We introduce a novel differentially private decentralized non-convex optimization algorithm that incorporates random agent activation and message sparsification~(\textbf{Algorithm~\ref{algo: one}}). 
Unlike previous studies that focused solely on communication reduction~\cite{yi2018dynamic, huo2024distributed, reisizadeh2019exact, wang2021error} or privacy preservation~\cite{huang2015differentially, ding2021differentially, liu2024distributed}, our algorithm simultaneously enhances both communication efficiency and privacy preservation. 
{Additionally, in contrast to~\cite{hu2023federated, chen2024privacy}, our approach does not require a central node for global information aggregation.}
\item[2)] We present a rigorous DP analysis of the designed algorithm~(\textbf{Theorem~\ref{thm: privacy}}), quantifying the amplified privacy level attributed to random activation and sparsification. 
Our findings indicate that reducing the agent activation probability and the number of coordinates during sparsification can lower the required noise for achieving $(\varepsilon, \delta)$-DP, unlike the works of Xie et al.~\cite{xie2023compressed, xie2023differentially} which overlook the interplay between compression and privacy preservation.
To our knowledge, this is the first study to analyze how efficient communication impacts privacy amplification in decentralized non-convex optimization.
\item[3)] We establish the non-asymptotic convergence of our algorithm in finding first-order stationary points and demonstrate a better trade-off between privacy, accuracy, and communication efficiency~(\textbf{Theorem~\ref{thm: convergence}}, \textbf{Corollary~\ref{cor: optimal}}).
\end{itemize}

The paper is organized as follows. Section~\ref{sec: preliminaries} lays out the preliminaries and problem formulations.
In Section~\ref{sec: algo}, we propose a decentralized optimization algorithm with amplified DP.
Section~\ref{sec: privacy} discusses the DP analysis and Section~\ref{sec: convergence} delves into the convergence analysis.
Section~\ref{sec: sim} presents numerical simulations to illustrate our findings. 
Finally, Section~\ref{sec: conclusions} covers conclusions and future research.

\emph{Notation:} 
Let $\mathbb{R}^{p}$ and $\mathbb{R}^{p \times q}$ denote $p$-dimensional vectors and $p \times q$-dimensional matrices, respectively.
$I_{p} \in \mathbb{R}^{p \times p}$ represents a $p \times p$-dimensional identity matrix, $\mathbf{1}_{p} \in \mathbb{R}^{p}$ is a $p$-dimensional vector, and $[d]$ denotes the set of integers $\{1, 2, \dots, d \}$.
The notation $\text{diag} \{a_{1}, a_{2}, \dots, a_{N}\}$ denotes a diagonal matrix with elements $a_{1}, a_{2}, \dots, a_{N}$.
Let $[\cdot]_{j}$ signify the $j$-th coordinate of a vector, $c$ represent a set of integers, and $[x]_{c}$ be a vector containing elements $[x]_{j}$ for $j \in c$.
Denote $\|\cdot\|$ as the $\ell_{2}$-norm for a vector and the induced-$2$ norm for a matrix, while $\|\cdot \|_{F}$ as the Frobenius norm for a matrix.
$\mathcal{N}(\mu, \sigma^{2})$ stands for the Gaussian distribution with expectation $\mu$ and standard deviation $\sigma$, and $\text{\ttfamily Bern}(p)$ denotes the Bernoulli distribution with probability $p$ for a value of $1$ and with probability $1-p$ for a value $0$.
We use $\mathbb{P}\{U\}$ to represent the probability of event $U$, and $\mathbb{E}[x]$ to be the expected value of random variable $x$.
The notation $O(\cdot)$ is used to describe the asymptotic upper bound. Mathematically, $h(n) = O(g(n))$ if there exist positive constants $C$ and $n_{0}$ such that $0 \leq h(n) \leq Cg(n)$ for all $n \geq n_{0}$. Similarly, the notation $\Omega(\cdot)$ provides the asymptotic lower bound, i.e., $h(n) = \Omega(g(n))$ if there exist positive constants $C$ and $n_{0}$ such that $0 \leq Cg(n) \leq h(n)$ for all $n \geq n_{0}$.

\section{Preliminaries and Problem Formulation} \label{sec: preliminaries}
This section presents the basic setup of decentralized non-convex optimization and a class of decentralized momentum SGD algorithms. We will discuss the potential privacy leakage in traditional algorithms and introduce some concepts of DP. Finally, we will formulate our problems.
\subsection{Basic Setup}
We consider the finite-sum stochastic optimization problem
\begin{equation} \label{eq: problem}
	\min_{x \in \mathbb{R}^{d}} \ f(x) = \frac{1}{n} \sum_{i=1}^{n} f_{i}(x), \ f_{i}(x) = \mathbb{E}_{\zeta_{i} \sim \mathcal{D}_{i}} l_{i}(x, \zeta_{i}),
\end{equation}
where $n$ is the number of agents, $\mathcal{D}_{1}, \dots, \mathcal{D}_{n}$ are distributions of the local dataset on every agent, $l_{i}: \mathbb{R}^{d} \times \mathcal{D}_{i} \to \mathbb{R}$ are possibly non-convex loss functions. 
This scenario encompasses the crucial scenario of minimizing empirical risk in distributed learning applications~\cite{liu2024distributed}.
We make the following assumptions on the loss function:
\begin{assumption} \label{assum: f} 
Each loss function $f_{i}: \mathbb{R}^{d} \to \mathbb{R}$ for $i \in [n]$ satisfies the following conditions.
\begin{itemize}
    \item[\romannumeral1)] The function $f_{i}$ is $L$-smooth, i.e., for any $x, y \in \mathbb{R}^{d}$,
	\begin{equation*}
	f_{i}(y) \leq f_{i}(x) + \left< \nabla f_{i}(x), y-x  \right> + \frac{L}{2}\|y-x \|^{2}.
	\end{equation*}
  \item[\romannumeral2)] The variance of stochastic gradients computed at $f_{i}$ is bounded, i.e.,  
  $\mathbb{E}_{\zeta_{i}} \left[ \| \nabla l_{i}(x, \zeta_{i}) - \nabla f_{i}(x) \|^{2} \right] \leq \varsigma_{i}^{2}$, for all $x\in \mathbb{R}^{d}$, where $\mathbb{E}_{\zeta_{i}}[\cdot]$ denotes the expectation over $\zeta_{i} \sim \mathcal{D}_{i}$.
\end{itemize}

\end{assumption}
\begin{assumption} \label{assum: bounded}
 For any data sample $\zeta_{i} \sim \mathcal{D}_{i}$, we have $\left |[\nabla l_{i}(x, \zeta_{i})]_{j} \right| \leq G/\sqrt{d}$ for all $x\in \mathbb{R}^{d}$, $i \in [n]$, and $j \in [d]$.
\end{assumption}
Assumption~\ref{assum: f} is standard in non-convex optimization literature~\cite{reddi2020adaptive, xu2023distributed}. Assumption~\ref{assum: bounded} characterizes the sensitivity of the gradient query
$\nabla l_{i}(x, \zeta_{i})$ and implies $\mathbb{E}_{\zeta_{i}}[\| \nabla l_{i}(x, \zeta_{i})\|^{2}] \leq G^{2}$, which can be enforced by the gradient clipping technique~\cite{chen2020understanding}.

\subsection{Decentralized Momentum SGD}
In decentralized networks, each agent is only allowed to communicate with its local neighbors defined by the network topology, given as a weighted graph $\mathcal{G} = ([n], \mathcal{E})$, with the edge set $\mathcal{E}$ representing the communication links along which messages can be exchanged. We assign a positive weight $0< w_{ij} < 1$ to every edge, and $w_{ij}=0$ for disconnected agents $(i,j) \notin \mathcal{E}$. The following assumption about the weight matrix $W$ is used in our work.
\begin{assumption} \label{assum: W}
	The weight matrix $W$ is symmetric ($W^{\top} = W$) and doubly stochastic ($W\mathbf{1} = \mathbf{1}$, $\mathbf{1}^{\top}W = \mathbf{1}^{\top}$), with eigenvalues $1 = |\lambda_{1}(W)|> | \lambda_{2}(W) | \geq \cdots \geq |  \lambda_{n}(W) |$ and spectral gap $\rho := 1-|\lambda_{2}(W)|\in(0,1]$. 
\end{assumption}
A classical choice of $W$ is $W = I - \frac{1}{\iota} L$, where $\iota \in (d_{\max}, \infty)$, $d_{\max}$ is the maximum degree of $\mathcal{G}$, and $L$ is a Laplacian matrix of $\mathcal{G}$. It will also be convenient to define $\phi := \|I -W \|_{2} \in [0,2]$.

In decentralized SGD algorithms, using momentum instead of the plain gradient in the iteration allows acceleration~\cite{yu2019linear}. At iteration $t$, each agent $i$ stores its private variable, denoted as $x_{i,t} \in \mathbb{R}^{d}$, and a momentum variable, $m_{i,t}$. Agent $i$ updates these variables as 
\begin{subequations} \label{eq: traditional}
\begin{align}
	\label{eq: tm_update} m_{i,t+1} =& g_{i,t} + \beta m_{i,t}, \\
	\label{eq: tx_update} x_{i,t+1} =& x_{i,t} - \alpha m_{i,t+1} + \gamma \sum_{j \in [n]} w_{ij}(x_{j,t} - x_{i,t}),
\end{align}
\end{subequations}
where $g_{i,t} = \nabla l_{i}(x_{i,t}, \zeta_{i,t})$ is the stochastic gradient of local loss over $x_{i,t}$ with $\zeta_{i,t}$ randomly sampled from $\mathcal{D}_{i}$, $0< \beta < 1$ is a momentum factor,
 $\alpha > 0$ and $\gamma > 0$ are the SGD step size and the consensus step size, respectively.

Note that directly transmitting the original local variable at each iteration might waste communication resources and cause privacy leakage.

\subsection{Privacy Leakage and DP}
The local datasets usually contain sensitive information about agents. If problem~\eqref{eq: problem} is solved in an insecure environment, the information leakage will pose a threat to personal and property security. 
Specifically, let us consider the worst case, i.e., the attacker knows the information of $\alpha$, $\beta$, $\gamma$, and $W$ in~\eqref{eq: traditional}. Once intercepting $x_{i,t}$ for $t \geq 0$, the attacker can infer
\begin{equation*}
	m_{i,t+1} = \frac{x_{i,t} - x_{i,t+1} - \gamma \sum_{j \in [n]} w_{ij}(x_{j,t} - x_{i,t})}{\alpha}.
\end{equation*}
Then, the attacker can learn the exact value of the gradient by $g_{i,t} = m_{i,t+1} - \beta m_{i,t}$.
With the learned $g_{i,t}$, even raw data like pixel-wise accuracy for images and token-wise matching for texts can be precisely reversely inferred by deep leakage from gradient attackers~\cite{zhu2019deep}. 

To design a privacy-preserving algorithm, we utilize the concept of DP in this work. 
It refers to the property of a randomized algorithm $\mathcal{A}$ that the presence or absence of an individual in a dataset cannot be distinguished based on the output of $\mathcal{A}$~\cite{dwork2006differential}.
Formally, we introduce the following definition of DP in the context of decentralized stochastic optimization.

\begin{definition} \label{defn: DP}
Consider a communication network $\mathcal{G}$, in which each node has its dataset $\mathcal{D}_{i}$. Let $\{s_{i,t}, i \in [n] \}$ denote the set of messages exchanged among the nodes at iteration $t$. A distributed algorithm $\mathcal{A}$ satisfies $(\varepsilon, \delta)$-DP during $T$ iterations, if for every pair of neighboring dataset $\mathcal{D} = \cup_{i \in [n]}\mathcal{D}_{i}$ and $\mathcal{D}^{\prime} = \cup_{i \in [n]}\mathcal{D}_{i}^{\prime}$, and for any set of possible outputs $\mathcal{O}$ during $T$ iterations we have
	\begin{equation} \label{eq: dp_defn}
	\begin{aligned}
		& \mathbb{P}\{ \{s_{i,t}, i \in [n], t \in [T] \} \in \mathcal{O} | \mathcal{D} \} \\
		\leq & e^{\varepsilon} \mathbb{P}\{ \{s_{i,t}, i \in [n], t \in [T] \} \in \mathcal{O} | \mathcal{D}^{\prime} \} + \delta.
	\end{aligned}
	\end{equation}
\end{definition}
Definition~\ref{defn: DP} states that the output distributions of neighboring datasets exhibit small variation. 
The factor $\varepsilon$ in~\eqref{eq: dp_defn} represents the upper bound of privacy loss by algorithm $\mathcal{A}$, and $\delta$ denotes the probability of breaking this bound. Therefore, a smaller $\varepsilon$ correspond to a stronger privacy guarantee.
Gaussian mechanism (GM) is a commonly employed technique to achieve $(\varepsilon, \delta)$-DP, as given in the following lemma.

\begin{lemma} \label{lem: GM}
(Gaussian Mechanism~\cite{balle2018improving}) 
A GM $\mathcal{M}$ for a vector-valued computation $r: \mathcal{D} \to \mathbb{R}^{d}$ is obtained by computing the function $r$ on the input data $\zeta \in \mathcal{D}$ and then adding random Gaussian noise perturbation $\nu \sim \mathcal{N}(0, \sigma^{2}I_{d})$ to the output, i.e., 
\begin{equation*}
	\mathcal{M} = r(\zeta) + \nu.
\end{equation*}	
The GM $\mathcal{M}$ is $\left(\frac{\sqrt{2 \log (1.25/ \delta) }\Delta}{\sigma}, \delta \right)$-DP for any two neighboring dataset $\mathcal{D}$ and $\mathcal{D}^{\prime}$, where $\Delta$ denotes the sensitivity of $r$, i.e., $\Delta = \sup_{\mathcal{D}, \mathcal{D}^{\prime}} \| r(\mathcal{D}) - r(\mathcal{D}^{\prime}) \|$. 
\end{lemma}
Lemma~\ref{lem: GM} indicates that achieving $(\varepsilon, \delta)$-DP requires adjusting the noise intensity based on the privacy guarantee $\varepsilon$ and $\delta$, as well as the sensitivity $\Delta$.

\subsection{Problem Statement} \label{subsec: pro}
In this paper, we aim to answer the following questions:
\begin{itemize}
\item[(a)] Can we design an algorithm to reduce communication costs as well as guarantee DP for decentralized optimization problems~\eqref{eq: problem}?
\item[(b)] To achieve $(\varepsilon, \delta)$-DP in the proposed algorithm, what level of noise intensity, i.e., $\sigma$, shall we add?
\item[(c)] What optimization accuracy of the proposed algorithm can we achieve?
\end{itemize}

\section{Algorithm Development} \label{sec: algo}
In this section, we develop an algorithm to answer problem (a) in Section~\ref{subsec: pro}.
The framework depicted in Fig.~\ref{fig: framework} is introduced to simultaneously enhance communication efficiency and amplify DP. We will outline the design of algorithms for each component within this framework.
\subsection{Noise Perturbation}
In terms of achieving DP for optimization methods, there are two common methods. The first type disturbs the output of a non-private algorithm, and the second type perturbs the gradient. The former involves recursively estimating the time-varying sensitivity of updates, which makes the propagation of DP noise and its effect on convergence difficult to quantify. 
Therefore, we adopt the latter approach in this work and introduce noise to perturb the stochastic gradient in~\eqref{eq: tm_update}, allowing more flexibility to conduct both privacy and convergence analysis.

By perturbing the gradient, the update of momentum becomes
\begin{equation*}
	m_{i,t+1} = g_{i,t} + \theta_{i,t} + \beta m_{i,t},
\end{equation*}
where $\theta_{i,t} \sim \mathcal{N}(0, \sigma^{2}I_{d})$ is the Gaussian noise and $\sigma$ is its standard deviation. The noise is independently drawn by agent $i$ at each iteration.

\subsection{Sparsified Communication}
To reduce the message size, we introduce the compression technique. Commonly used compressors include quantizers and sparsifiers. 
Quantizers involve reducing the precision of messages by mapping them to a smaller set of discrete values with fewer bits, 
like random $b$-bits quantizers~\cite{huo2024compression},
while sparsifiers focus on transmitting partial elements of variables and setting the rest to zero, like $\text{\ttfamily Random}_{k}$~\cite{hu2023federated}, $\text{\ttfamily Top}_{k}$~\cite{yi2022communication}, and Bernoulli sparsifier~\cite{chen2024privacy}.  
Lemma~\ref{lem: GM} indicates that the required variance of Gaussian noise depends on the sensitivity of queries. 
Comparing quantizers with sparsifiers, we find that sparsifiers have the potential to diminish the sensitivity of messages, thereby amplifying the DP guarantee. 
Additionally, quantization would destroy the distribution of injected Gaussian perturbation, rendering privacy analysis intractable.
Therefore, we use sparsifiers in our work to enhance privacy preservation. 

The $\text{\ttfamily Random}_{k}$ and Bernoulli sparsifiers both randomly determine which coordinates to be kept in transmission, but they may discard coordinates that are actually important, which inevitably degrades model accuracy when $k$ is small. On the other hand, $\text{\ttfamily Top}_{k}$ keeps the coordinates with the largest magnitude and can achieve higher model accuracy for small $k$. Hence, we utilize the $\text{\ttfamily Top}_{k}$ sparsifier in our algorithm, which is formally defined as follows:
\begin{definition}
	($\text{\ttfamily Top}_{k}$ Sparsification): Given $x \in \mathbb{R}^{d}$ and a parameter $k \in [d]$, the largest $k$ coordinates in magnitude is selected among the vector $x$, i.e., 
	\begin{equation*}
		\mathcal{S}(x) := x \odot e ,
	\end{equation*}
	where $\odot$ denotes the Hadamard (element-wise) product, the element of $x$ is reordered as $|[x]_{1}| \geq |[x]_{2}| \geq \cdots |[x]_{d}|$, and $e$ satisfies $[e]_{l} = 1$ for $l \leq k$ and $[e]_{l} = 0$ for $l > k$.
\end{definition}

However, the mean square error of sparsification typically scales with the norm of the input vectors. 
Since the local states can be any value, sparsification can result in a large mean square error, which is difficult to handle. 
To address this issue, we introduce local copies of agents and treat the true local states as reference signals. 
By carefully designing appropriate step sizes, we guide these local copies toward the true local states, ensuring their differences remain bounded. 
Sparsifying this difference allows us to effectively manage the mean square error and achieve improved convergence results.
Specifically, we let each agent $i \in [n]$ store and update its own local variable $x_{i,t}$ as well as the variables $\hat{x}_{j,t}$ for all neighbors $j: (i,j) \in \mathcal{E}$. The communication algorithm can be summarized as
\begin{subequations} \label{eq: communication}
	\begin{align}
	x_{i,t+1} =& x_{i,t} - \alpha m_{i,t+1} + \gamma \sum_{j \in [n]} w_{ij}(\hat{x}_{j,t} - \hat{x}_{i,t}), \\
	s_{i,t} =& \mathcal{S}(x_{i,t+1} - \hat{x}_{i,t}), \\
	\hat{x}_{j,t+1} =& s_{j,t} + \hat{x}_{j,t}, \quad \forall j: (i,j) \in \mathcal{E}.
\end{align}
\end{subequations}
The agents communicate the sparsified updates $s_{i,t}$ with neighbors and update the variable $\hat{x}_{j,t}$ for all their neighbors. These $\hat{x}_{i,t}$ are available to all the neighbors of node $i$ and represent the ``publicly available" copies of the private $x_{i,t}$. It is worth noting that $\hat{x}_{i,t} \neq x_{i,t}$ generally due to the sparsification.

\subsection{Random Activation}
To further reduce the communication cost of the network and introduce randomness to amplify privacy, we activate each node randomly at each iteration. 
If node $i$ is active, it will update its variables using a stochastic gradient and send out the sparsified update. Otherwise, it will update without SGD and stay idle in transmission.

Denote $\eta_{i,t}$ to indicate whether agent $i$ transmits data at time $t$, i.e., 
\begin{equation*}
	\eta_{i,t} = \left\{
	\begin{array}{cl}
		1, & \text{agent $i$ is active in transmission}, \\
		0, & \text{agent $i$ is inactive in transmission}, \\
	\end{array}
	\right.
\end{equation*}
where $\mathbb{P}[\eta_{i,t}=1]=p$ and $\mathbb{P}[\eta_{i,t}=0]=1-p$ with $\frac{1}{2} \leq p \leq 1$. Then, the updates of local variables become
\begin{subequations} \label{eq: algo_3}
	\begin{align}
	\label{eq: m_update} m_{i,t+1} =& \eta_{i,t}(g_{i,t} + \theta_{i,t}) + \beta m_{i,t}, \\
	\label{eq: x_update} x_{i,t+1} =& x_{i,t} - \eta_{i,t}\left( \alpha m_{i,t+1} \right) + \gamma \sum w_{ij}(\hat{x}_{j,t} - \hat{x}_{i,t}), \\
 \label{eq: s_update} s_{i,t} =&  \mathcal{S}(x_{i,t+1} - \hat{x}_{i,t}) , \\
	\label{eq: replica_update} \hat{x}_{j,t+1} =& \eta_{j,t}s_{j,t} + \hat{x}_{j,t}, \quad \forall j:(i,j) \in \mathcal{E} .
\end{align}
\end{subequations}

In summary, our proposed framework in Fig.~\ref{fig: framework} underpins the algorithm designed to amplify DP through efficient communication, as detailed in Algorithm~\ref{algo: one}, called DO-ADP.
{At each iteration, DO-ADP uses communication resources at a rate of $pk/d \leq 100\%$, reducing communication costs via partial agent activation and sparsification. 
In the following section, we will analyze how DO-ADP leverages the features from the designed communication reduction scheme to enhance DP preservation. 
}

\begin{figure}[t]
  \begin{algorithm}[H]
	\caption{Decentralized Optimization with Amplified Differential Privacy (DO-ADP)}
	\begin{algorithmic}[1] \label{algo: one}
	\renewcommand{\algorithmicrequire}{\textbf{Input:}}
    \renewcommand{\algorithmicensure}{\textbf{Initialize:}}
    \REQUIRE Public information $W$, consensus step size $\gamma$, SGD step size $\alpha$, noise variance $\sigma$, momentum factor $\beta$, and the total number of iterations $T$.
    \ENSURE  $x_{i,0} = x_{0} \in \mathbb{R}^{n}$, $\hat{x}_{i,0} = \mathbf{0}$, and $m_{i,0} = \mathbf{0}$.
	\FOR {$k = 0, 1, 2,  \dots, T-1$}
	\FOR{for each $i \in \mathcal{N}$}
\STATE Draw $\eta_{i,t} \sim \text{\ttfamily Bern}(p)$.
\IF{node $i$ is active ($\eta_{i,t}=1$)}
\STATE Sample $\zeta_{i,t}$ uniformly from local dataset $\left\{\zeta_{i}^{(1)}, \dots, \zeta_{i}^{(q)} \right \}$ and compute the gradient $g_{i,t} = \nabla f_{i}(x_{i,t}, \zeta_{i,t})$.
\STATE Update the local momentum with the gradient and weight decay
\begin{equation*}
	m_{i,t+1} = g_{i,t} + \theta_{i,t} + \beta m_{i,t},
\end{equation*}
where $\theta_{i,t} \sim \mathcal{N}(0, \sigma^{2})$.
\STATE $x_{i,t+1} = x_{i,t} - \alpha m_{i,t+1} + \gamma \sum w_{ij}(\hat{x}_{j,t} - \hat{x}_{i,t})$.
\STATE Sparsifies the update: $s_{i,t} = \mathcal{S}(x_{i,t+1} - \hat{x}_{i,t})$.
\STATE Send $s_{i,t}$.
\ELSE
\STATE $m_{i,t+1} = \beta m_{i,t}$,
\STATE $x_{i,t+1} = x_{i,t} + \gamma \sum w_{ij}(\hat{x}_{j,t} - \hat{x}_{i,t})$.
\ENDIF
\STATE Receives $s_{j,t}$ from active neighbors.
\STATE Updates the local replica $\hat{x}_{j,t+1} = s_{j,t} + \hat{x}_{j,t}$ for active neighbors and $\hat{x}_{k,t+1} = \hat{x}_{k,t}$ for inactive neighbors.
    \ENDFOR
    \ENDFOR
	\end{algorithmic}
  \end{algorithm}
\end{figure}

\section{Privacy Analysis} \label{sec: privacy}
In this section, we provide a privacy guarantee of Algorithm~\ref{algo: one} and address question (b) stated in Section~\ref{subsec: pro}.

To analyze the effect of sparsification communication on privacy performance, let $c_{i}^{t}$ denote the selected coordinate set for participating agent $i$ via $\text{\ttfamily Top}_{k}$ at round $t$, i.e., $\mathcal{S}(\cdot) = [\cdot]_{c_{i}^{t}}$. Therefore, the transmitted message is
\begin{equation*}
\begin{aligned}
	s_{i,t} =& [x_{i,t+1} - \hat{x}_{i,t}]_{c_{i}^{t}} =   [x_{i,t+1}]_{c_{i}^{t}} -  [\hat{x}_{i,t}]_{c_{i}^{t}}\\
	=& \left[x_{i,t} -\eta_{i,t} \alpha \beta m_{i,t} + \gamma \sum w_{ij}(\hat{x}_{j,t} - \hat{x}_{i,t})\right]_{c_{i}^{t}} \\
	& \quad - \alpha \eta_{i,t}  \left[ g_{i,t} + \theta_{i,t}\right]_{c_{i}^{t}} -  [\hat{x}_{i,t}]_{c_{i}^{t}}.
\end{aligned}
\end{equation*}
An important observation is that only the value in $c_{i}^{t}$ are transmitted among neighboring agents. Hence, the privacy level will remain the same if we only add noise to the gradient with the selected coordinate.
In other words,
what matters in privacy protection is the sparsified noise gradient update, which can be represented as $\left[ g_{i,t} + \theta_{i,t}\right]_{c_{i}^{t}}
	= \left[g_{i,t}\right]_{c_{i}^{t}} + \left[\theta_{i,t} \right]_{c_{i}^{t}}$.
Let $\Delta$ be the $\ell_{2}$-sensitivity of $\left[g_{i,t}\right]_{c_{i}^{t}}$, then we have
\begin{align} \label{eq: sensivity}
	\Delta^{2} 
	=& \max_{\mathcal{D}_{i}, \mathcal{D}_{i}^{\prime}} \left \| \left[\hat{g}_{i,t}\right]_{c_{i}^{t}} - \left[\hat{g}_{i,t}^{\prime}\right]_{c_{i}^{t}} \right \|^{2} \nonumber \\
	=& \max_{\mathcal{D}_{i}, \mathcal{D}_{i}^{\prime}} \left\|  \left[\nabla l_{i}(x_{i,t}, \zeta_{i,t}) - \nabla l_{i}(x_{i,t}, \zeta_{i,t}^{\prime}) \right]_{c_{i}^{t}} \right \|^{2} \nonumber \\
	\leq & \frac{4k G^{2} }{d},
\end{align}
where the last inequality holds from Assumption~\ref{assum: bounded}.
Compared with traditional optimization algorithms with DP where the sensitivity is $2G$~\cite{huang2015differentially, liu2024distributed}, the sensitivity in our algorithm is reduced by a ratio of $\sqrt{k/d}$ due to the sparsification, thereby reducing the required noise intensity.

Based on the sensitivity in~\eqref{eq: sensivity} and some useful properties of DP shown in Lemmas~4--6 in Appendix~A, we have the privacy guarantee for Algorithm~\ref{algo: one}.
\begin{theorem} \label{thm: privacy}
Suppose Assumption~\ref{assum: bounded} holds. Given parameters $\varepsilon \in (0,1]$, and $\delta_{0} \in (0,1]$, if
	\begin{equation} \label{eq: noise}
		\sigma^{2} \geq \frac{160 k p^{2} T \log(1.25/\delta_{0})G^{2}}{ q^{2} d \varepsilon^{2} }
	\end{equation}
	and $T \geq \frac{ q^{2} \varepsilon^{2}}{4 p^{2}}$, then Algorithm~\ref{algo: one} is $(\varepsilon, \delta)$-DP for some constant $\delta \in (0,1]$.
\end{theorem}
\begin{proof}
We first analyze the DP at each local iteration.
The noise added to each coordinate of $[g_{i,t}]_{c_{i}^{t}} = [\nabla l_{i}(x_{i,t}, \zeta_{i,t})]$ is drawn from the Gaussian distribution $\mathcal{N}(0, \sigma^{2})$. Then based on~\eqref{eq: sensivity} and Lemma~\ref{lem: GM}, $[g_{i,t}]_{c_{i}^{t}} + [\theta_{i,t}]_{c_{i}^{t}}$, denoted as $\mathcal{M}_{t}$, achieves $(\varepsilon_{t}, \delta_{0})$-DP with 
\begin{equation*}
	\varepsilon_{t} = \frac{2\sqrt{2k \log (1.25/\delta_{0})}G}{\sigma  \sqrt{d}}
\end{equation*} 
for any $\delta_{0} \in [0,1]$. Due to the conditions on $\sigma$ and $T$, we obtain
\begin{equation} \label{eq: eps_t}
	\varepsilon_{t}^{2} = \frac{8k\log(1.25/\delta_{0})C^{2}}{\sigma^{2}d} \leq \frac{q^{2}\varepsilon^{2}}{20p^{2} T} \leq \frac{1}{5} .
\end{equation}
Recall from Algorithm~\ref{algo: one} that, at each iteration, each agent is randomly sampled with probability $p$, and each active agent randomly selects a data sample from $q$ instances to compute stochastic gradients.
Denote $\mathcal{A}_{t}$ the composition of $\mathcal{M}_{t}$ and the subsampling (including random activation and sparsification) procedure. Upon using Lemma~4 in Appendix~A and~\eqref{eq: eps_t}, we obtain that $\mathcal{A}_{t}$ is $(\varepsilon_{t}^{\prime}, p\delta_{0}/q)$-DP with
\begin{equation*}
	\varepsilon_{t}^{\prime} = \frac{2 p \varepsilon_{t}}{q} \geq \frac{p(e^{\varepsilon_{t}}-1)}{q} \geq \ln\left(1+ \frac{p(e^{\varepsilon_{t}}-1)}{q} \right).
\end{equation*} 
Additionally, because of~\eqref{eq: eps_t}, we get $\varepsilon_{t}^{\prime} = 2 p \varepsilon_{t}/q \leq 2 \varepsilon_{t} \leq 0.9$ and 
\begin{equation*}
	\sum_{t=1}^{T} \varepsilon_{t}^{\prime 2} \leq \frac{4 p^{2}}{q^{2}} \sum_{t=1}^{T} \frac{8k\log(1.25/\delta_{0})C^{2}}{\sigma^{2}d} = \frac{1}{5} \sum_{t=1}^{T}\frac{\varepsilon^{2}}{T} \leq 1.
\end{equation*}

Then, after $T$ iterations, we consider the composition of $\mathcal{A}_{1}, \dots, \mathcal{A}_{T}$, denoted by $\mathcal{A}$. 
Based on the advanced composition rule for DP in Lemma~5, we obtain $\mathcal{A}$ is $(\tilde{\varepsilon}, \tilde{\delta})$-DP with
\begin{equation*}
	\tilde{\varepsilon} = \sqrt{\sum_{t=1}^{T}2\varepsilon_{t}^{ \prime 2} \log \left( e + \frac{\sqrt{\sum_{t=1}^{T} \varepsilon_{t}^{\prime 2}}}{\delta^{\prime}} \right) } + \sum_{t=1}^{T} \varepsilon_{t}^{ \prime 2}
\end{equation*}
and $\tilde{\delta} = 1-(1-\delta^{\prime})(1-p \delta_{0})^{T}$ for any $\delta^{\prime} \in (0,1]$. By setting $\delta^{\prime} = \sqrt{\sum_{t=1}^ {T} \varepsilon_{t}^{\prime 2}}$, there holds
\begin{align*}
	\tilde{\varepsilon} =& \sqrt{\sum_{t=1}^{ T}2\varepsilon_{t}^{\prime 2} \log \left( e + 1 \right) } + \frac{1}{5}\varepsilon^{2} \\
	\leq & \sqrt{3\sum_{t=1}^{T}\varepsilon_{t}^{\prime 2}} + \frac{1}{5} \varepsilon = \sqrt{\frac{3}{5}\varepsilon^{2}} + \frac{1}{5}\varepsilon \leq  \varepsilon,
\end{align*}
By setting $\delta =\tilde{\delta}$, we demonstrate that $\mathcal{A}$ is $(\varepsilon, \delta)$-DP.

The transmitted messages $\{ s_{i,t}, i \in [n], t \in [T] \}$ are computed based on the output of $\mathcal{A}$, i.e., perturbed gradients. By the post-processing property of DP in Lemma~6, Algorithm~\ref{algo: one} also satisfies $(\varepsilon, \delta)$-DP specified in Definition~\ref{defn: DP}.
\end{proof}
\begin{remark}
	In traditional differential private distributed optimization algorithms, the required Gaussian noise variance typically is $\Omega \left( \frac{T\log(1.25/\delta_{0})G^{2}}{q^{2}\varepsilon^{2}} \right)$~\cite{huang2019dp}. We reduced the required noise to $\Omega \left( \frac{k p^{2} T\log(1.25/\delta_{0})G^{2}}{q^{2}d\varepsilon^{2}} \right)$ with a ratio $\frac{kp^{2}}{d}$, where $k/d$ is from the $\text{\ttfamily Top}_{k}$ sparsification and $p^{2}$ is from the random activation. This result indicates that the required noise can be further reduced without compromising the DP guarantee by decreasing the number of selected coordinates, $k$, and activation probability $p$, which is consistent with our intuition since fewer communicated messages lower the risk of privacy leakage. 
\end{remark}

\begin{remark}
   In contrast to the findings in Xie et al.~\cite{xie2023compressed, xie2023differentially}, which do not demonstrate any impact of compressed communication on privacy preservation and optimization accuracy, we employ random activation for subsampling and use sparsification to decrease the sensitivity of the algorithm. 
   Leveraging these characteristics of efficient communication, the proposed algorithm reduces the noise intensity for certain privacy levels and better balances privacy and accuracy.
\end{remark}

\section{Convergence Analysis} \label{sec: convergence}
In this section, we will provide the convergence analysis of Algorithm~\ref{algo: one}. 
We first bound the momentum variable for agent $i$ at each iteration.
\begin{lemma} \label{lem: m}
	Under Assumptions~\ref{assum: f} and~\ref{assum: bounded}, for the sequence $\{m_{i,t}, i \in [n], t \geq 1 \}$ generated by Algorithm~\ref{algo: one}, we have that
	\begin{equation*}
		\mathbb{E}\left[ \| m_{i,t} \|^{2} \right] \leq \frac{ p( G^{2} + \sigma^{2}d)}{(1-\beta)^{2}}, \ i \in [n], \ t \geq 1.
	\end{equation*}
\end{lemma}
\begin{proof}
	The proof is provided in Appendix~B.
\end{proof}

Then, we can derive the upper bound of the mean square distance between the local optimization variables $x_{i,t}, i \in [n]$ and their average $\bar{x}_{t} = \frac{1}{n}\sum_{i=1}^{n} x_{i,t}$.
\begin{lemma} \label{lem: x}
Under Assumptions~\ref{assum: f}--\ref{assum: W}, for the sequence $\{ x_{i,t}, i \in [n], t \in [T]\}$ generated by Algorithm~\ref{algo: one}, we have that for all $t \geq 1$
\begin{align*}
	\sum_{i=1}^{n}\mathbb{E}\left[ \|\bar{x}_{t} - x_{i,t} \|^{2} \right] \leq \frac{8 p \alpha^{2}(G^{2}+ \sigma^{2}d)n}{c^{2}(1-\beta)^{2}}
\end{align*}
where $c = \frac{\rho^{2}pk}{82d}$.
\end{lemma}
\begin{proof}
	The proof is provided in Appendix~C.
\end{proof}

According to Lemmas~\ref{lem: m} and~\ref{lem: x}, we have the following convergence results for first-order stationary points.
\begin{theorem} \label{thm: convergence}
	Under Assumptions~\ref{assum: f}--\ref{assum: W}, if we choose $\gamma = \frac{\rho p k}{d(16\rho + \rho^{2} + 4 \phi^{2} + 2\rho \phi^{2}) - 8\rho p k}$ and $\alpha < \frac{(1-\beta)^{2}}{2L}$ for given $0< \beta < 1$, we have
	\begin{align*}
	 & \frac{1}{T} \sum_{t=0}^{T-1} \mathbb{E}\left[ \left\| \nabla f(\bar{x}_{t}) \right\|^{2} \right] \\
\leq & \frac{2(1-\beta)(f(x_{0}) - f^{*})}{\alpha p T} \\
&+ \frac{\alpha L(\beta + 2 + 4 \beta^{2}p)(G^{2} + \varsigma^{2} + \sigma^{2}d)}{n(1-\beta)^{3}} \\
& + \frac{8\alpha^{2}p(G^{2} + \sigma^{2}d)L^{2}}{(1-\beta)^{2}c^{2}}.
\end{align*}
\end{theorem}
\begin{proof}
	The proof is provided in Appendix~D.
\end{proof}

\begin{corollary} \label{cor: optimal}
Under Assumptions~\ref{assum: f}--\ref{assum: W}, if we choose $\gamma = \frac{\rho p k}{d(16\rho + \rho^{2} + 4 \phi^{2} + 2\rho \phi^{2}) - 8\rho p k}$ for given $0< \beta < 1$, there exists a constant step size $\alpha$ such that 
\begin{equation} \label{eq: con_1}
\begin{aligned}
     & \frac{1}{T} \sum_{t=0}^{T-1} \mathbb{E}\left[ \left\| \nabla f(\bar{x}_{t}) \right\|^{2} \right] \\
\leq & O \left(   \left( \frac{(1+p)(1+\sigma^{2}d)}{npT}  \right)^{\frac{1}{2}} +  
\left(   \frac{\sqrt{p(1+\sigma^{2}d)}}{cpT}  \right)^{\frac{2}{3}} + \frac{1}{pT}
\right).
\end{aligned}
\end{equation} \label{eq: con_2}
Furthermore, if we choose $\alpha = \sqrt{\frac{n}{T}}$ for $T \geq \frac{4nL^{2}}{(1-\beta)^{4}}$, there is
\begin{equation}
\begin{aligned}
    & \frac{1}{T} \sum_{t=0}^{T-1} \mathbb{E}\left[ \left\| \nabla f(\bar{x}_{t}) \right\|^{2} \right] \\
\leq & \frac{2(f(x_{0}) - f^{*})(1-\beta)}{p\sqrt{nT}} + \frac{ L(\beta + 2 + 4 \beta^{2}p)(G^{2} + \varsigma^{2} + \sigma^{2}d) }{(1-\beta)^{3}\sqrt{nT}} \\
&+ \frac{8np(G^{2} + \sigma^{2}d)L^{2}}{(1-\beta)^{2}c^{2}T}.
\end{aligned}
\end{equation}

\end{corollary}
\begin{proof}
    The proof is provided in Appendix~G.
\end{proof}

\begin{remark}
The first term in~\eqref{eq: con_1} indicates a linear speed-up compared to SGD on a single node. The noise level $\sigma$ impacts both lower and higher order terms, whereas sparsification affects only the higher order term. Therefore, the benefits of noise reduction would outweigh the impacts of information incompletion, thereby improving the trade-off between accuracy, privacy, and communication efficiency.
Additionally, as observed from~\eqref{eq: con_1}, a lower activation probability leads to slower convergence.
\end{remark}

\section{Numerical Experiments} \label{sec: sim}
In this section, we evaluate the performance of DO-ADP using two benchmark tests.
First, we apply DO-ADP to solve a decentralized logistic regression problem on the \emph{epsilon} dataset~\cite{sonnenburg2006large}, which is a distributed strongly convex problem. 
We compare our results with state-of-the-art algorithms to highlight the advantages of DO-ADP. 
Next, we use DO-ADP to tackle a decentralized non-convex problem by training a convolutional neural network (CNN) model on the MNIST dataset~\cite{deng2012mnist}. This demonstrates the algorithm's effectiveness in handling non-convex problems.
Throughout the experiments, we consider a graph 
$\mathcal{G} = ([n], \mathcal{E})$ with $\mathcal{E} = \{ (i, (i \pm j) \ \text{mod} \ n), | i \in [n], j=1, 2, 3 \}$, 
where $\text{mod}$ denotes the modulo operation. 
The data samples are randomly and evenly distributed among $n=20$ agents.

\subsection{Decentralized Logistic Regression}
The \emph{epsilon} dataset consists of 400,000 samples, each with 2,000 features.
Agent $i$ has local cost function
\begin{align*}
f_{i}(x) = \frac{1}{q} \sum_{j=1}^{q} \log (1+ \exp(-b_{j}a_{j}^{\top}x)) + \frac{1}{2q}\|x \|^{2},
\end{align*}
where $a_{j} \in \mathbb{R}^{d}$ and $b_{j} \in \{-1,1 \}$ are the data samples, and $q$ is the number of samples in the local dataset. The ground truth $f(x^{*})$ is obtained using LogisticSGD optimizer from scikit-learn~\cite{pedregosa2011scikit}.
In iteration $t$, agent $i$ randomly samples a data point form $\{\zeta_{i}^{(1)} , \dots, \zeta_{i}^{(q)}\}$ and updates its states using DO-ADP. The step sizes are $\alpha = 0.001$, $\gamma = 0.05$, and $\beta = 0.15$.
{One epoch is completed after sampling $q$ times.} 
We run the algorithm for six epochs and repeat each experiment five times.
\begin{figure}[t]
	\centering
	\begin{subfigure}[t]{0.35\linewidth}
		\includegraphics[width=\linewidth]{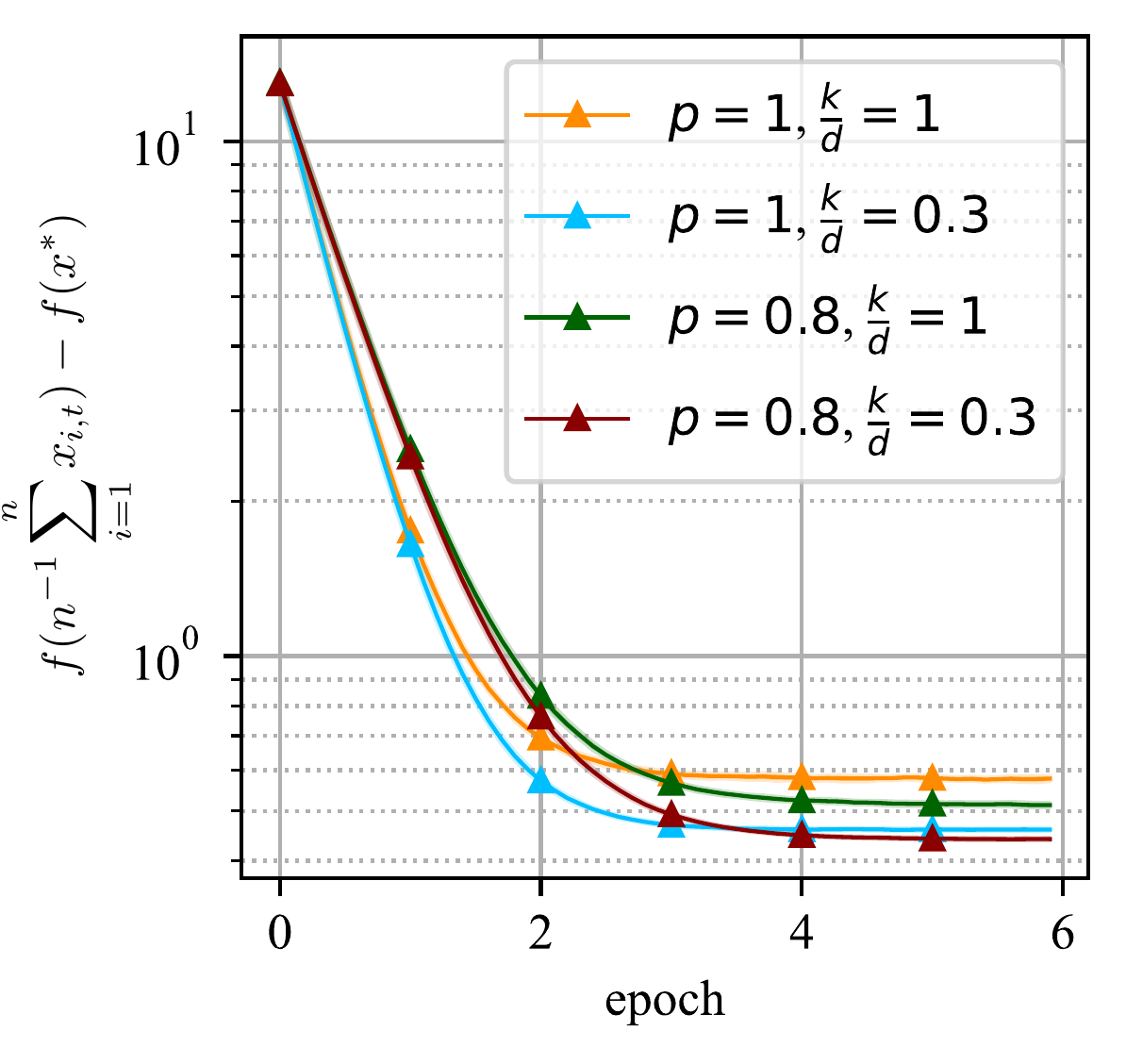}
  \captionsetup{font = small}
		\caption{Convergence under $\varepsilon=0.01$.}
		\label{fig: convergence}
	\end{subfigure}
	\begin{subfigure}[t]{0.35\linewidth}
		\includegraphics[width=\linewidth]{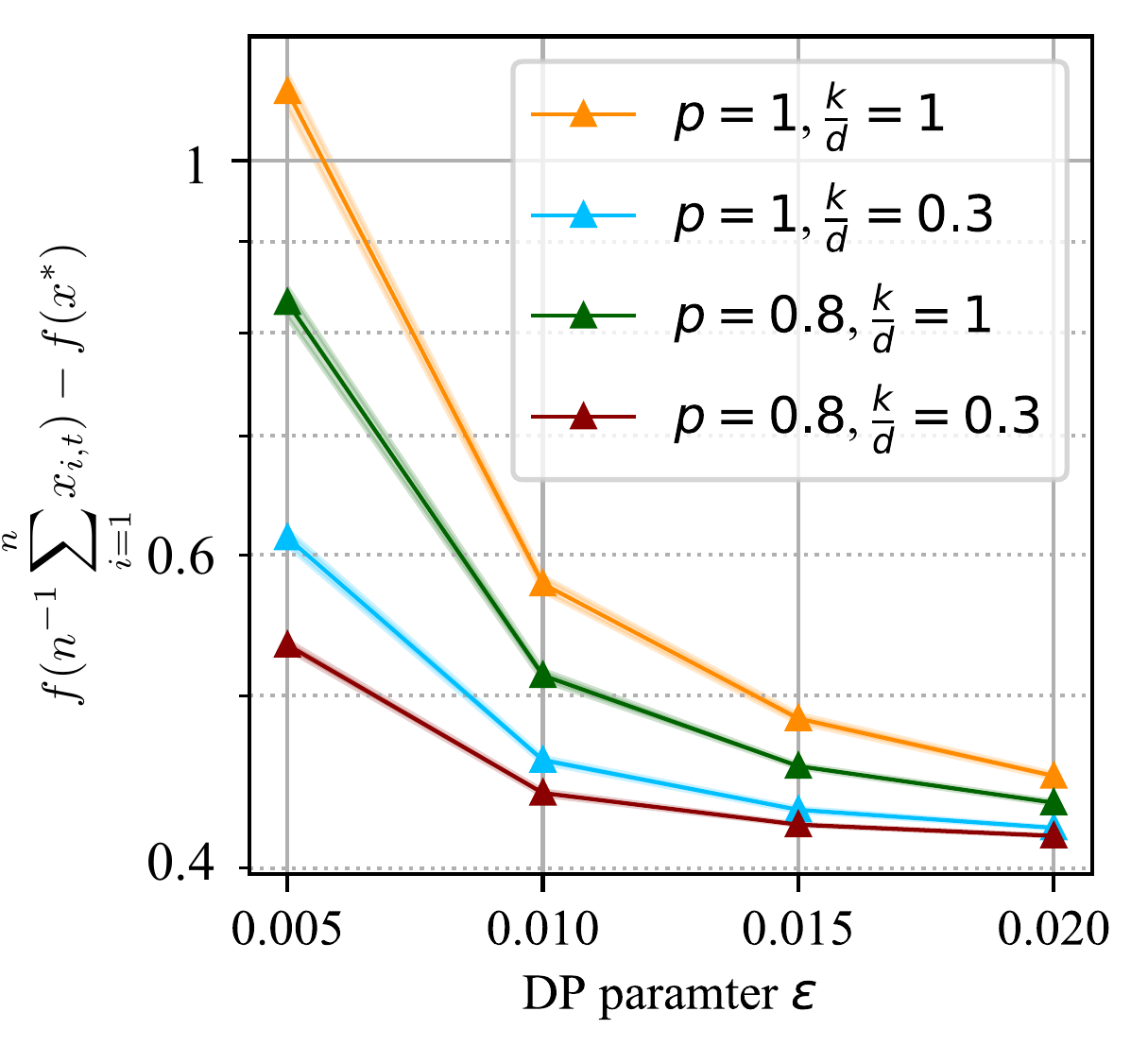}
  \captionsetup{font = small}
		\caption{Suboptimality under different DP parameters.} 
		\label{fig: tradeoff} 
	\end{subfigure}
\caption{Performance of DO-ADP.
{During each iteration, the communication resource utilization rates for the yellow, green, blue, and red lines are $100\%$, $80\%$, $30\%$, and $24\%$, respectively.}}
\label{fig: com}
\end{figure}
Fig.~\ref{fig: com} indicates that a higher sparsification ratio $k/d$ and activation probability $p$, result in a higher communication resource utilization rate and larger optimization error, which is consistent with Corollary~\ref{cor: optimal}. 
This suggests that, in this specific example, reduced noise may have greater benefits on convergence performance than the negative impacts of reduced information completeness.
Fig.~\ref{fig: tradeoff} illustrates the trade-off between privacy and optimization accuracy.
It shows that a higher value of $\varepsilon$, representing a less stringent privacy requirement, leads to reduced optimization error.
Moreover, the performance gap between different communication resource utilization rates is more significant when $\varepsilon$ is smaller, i.e., a tighter DP requirement.

We conduct a comparison between DO-ADP and several existing algorithms. 
The differentially private distributed optimization algorithm (DPDO) in~\cite{huang2015differentially} uses diminishing Laplacian noise to preserve DP.
The differentially private distributed dual average algorithm (DP-DDA) in~\cite{liu2024distributed} employs node sampling to enhance privacy, where inactive nodes do not update using gradient but still transmit data to neighbors. Since all agents transmit $d$-dimensional vectors per iteration in DPDO and DP-DDA, these algorithms utilize $100\%$ of communication resources per iteration. 
The compressed differentially private distributed gradient tracking algorithm (CPGT) in~\cite{xie2023compressed, xie2023differentially} adds Laplacian noise to data and then transmits compressed versions.
Table~\ref{tab: compare} offers a brief comparison of these existing algorithms and DO-ADP.
\begin{table}[t] 
	\caption{Comparison of different algorithms.}
	 \centering
  \footnotesize
	 \begin{tabular}{|m{0.25\linewidth}<{\centering}|m{0.25\linewidth}<{\centering}|m{0.25\linewidth}<{\centering}|}
	 	\hline
	 	Algorithm & Communication saving & Privacy amplification \\
	 	\hline
	 	DPDO~\cite{huang2015differentially} & $\times$   & $\times$ \\
	 	\hline
	 	DP-DDA~\cite{liu2024distributed} 
   & $\times$ & $\checkmark$ \\
	 	\hline
	 	CPGT~\cite{xie2023compressed} 
   & $\checkmark$ & $\times$ \\
	 	\hline
	 	DO-ADP & $\checkmark$  & $\checkmark$ \\
	 	\hline
	 \end{tabular}
  \label{tab: compare}
\end{table}
Fig.~\ref{fig: soa_svm} depicts the comparison results at the same privacy level, $\varepsilon = 0.01$.
\begin{figure}[t]  
	\centering
	\includegraphics[width=0.4\linewidth]{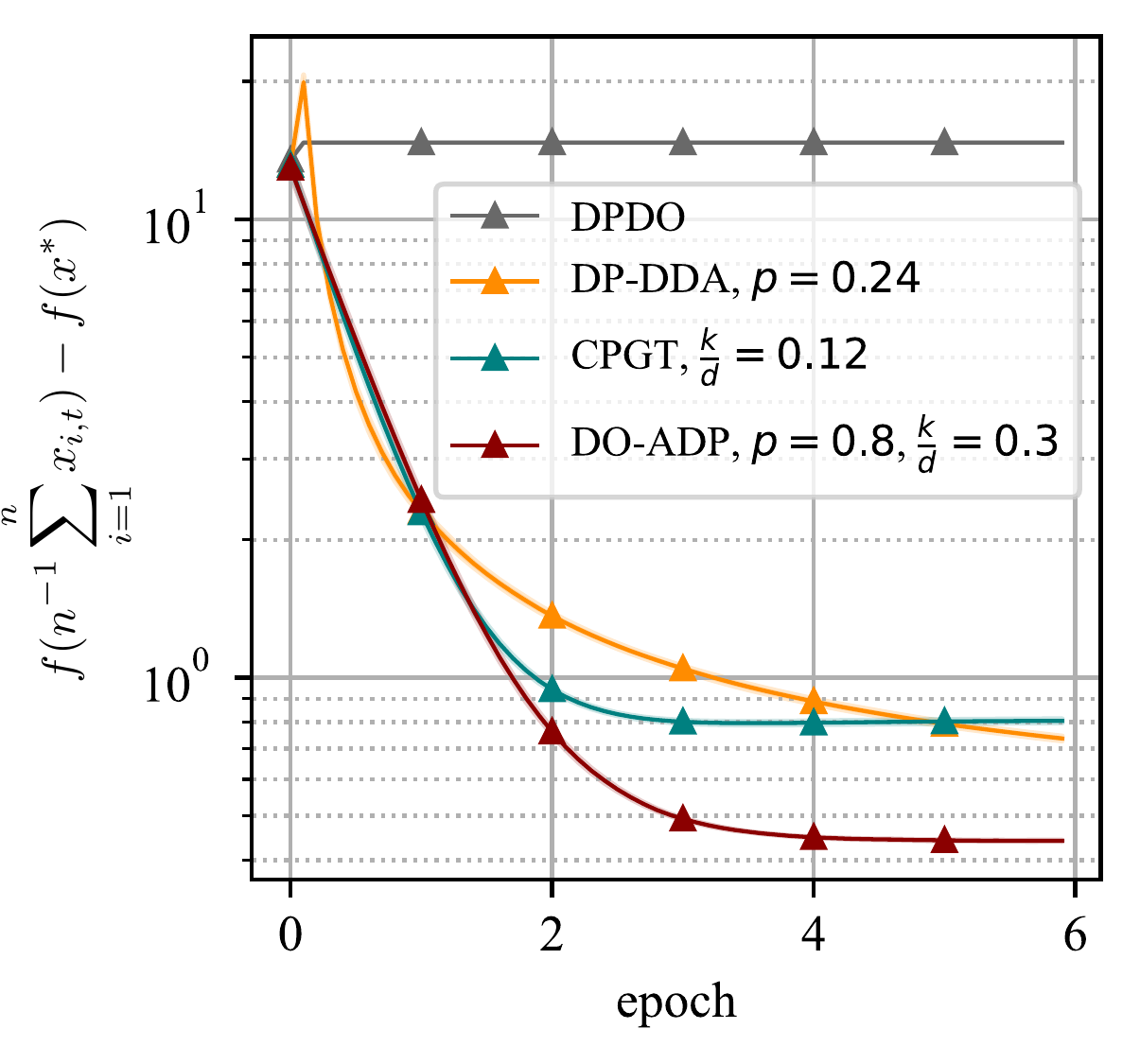}
	\caption{Convergence performance of DPDO, DP-DDA, GPGT, and DO-ADP with $\varepsilon = 0.01$.
{GPGT and DO-ADP utilize $24\%$ of communication resources per iteration, while DPDO and DP-DDA use $100\%$.}
 } 
	\label{fig: soa_svm}
\end{figure}
In DO-ADP, we set $p=0.8$ and $k/d=0.3$, allowing agents to use only  $24\%$ communication resources, hence reducing noise intensity by $24\%$. 
For fairness, we set the node subsampling ratio in DP-DDA to $24\%$, ensuring the same noise reduction. 
Since CPGT transmits two variables per iteration, we set its sparsification ratio as $k/d=0.12$ to match our resource usage. 
Fig.~\ref{fig: soa_svm} shows that our algorithm achieves the best optimization accuracy. 

\subsection{Decentralized CNN Training}
The MNIST dataset has 60,000 binary images for training and 10,000 binary images for testing. 
Each agent has a local copy of the CNN model with three conventional layers. 
We set $\alpha = 0.01$, $\gamma = 0.05$, and $\beta = 0.15$. 
The training and testing results are shown in Fig.~\ref{fig: cnn} and Table~\ref{tab: CNN_trade}, which demonstrates our theoretical results in the non-convex scenario.
\begin{figure}[t]  
	\centering
	\includegraphics[width=0.6\linewidth]{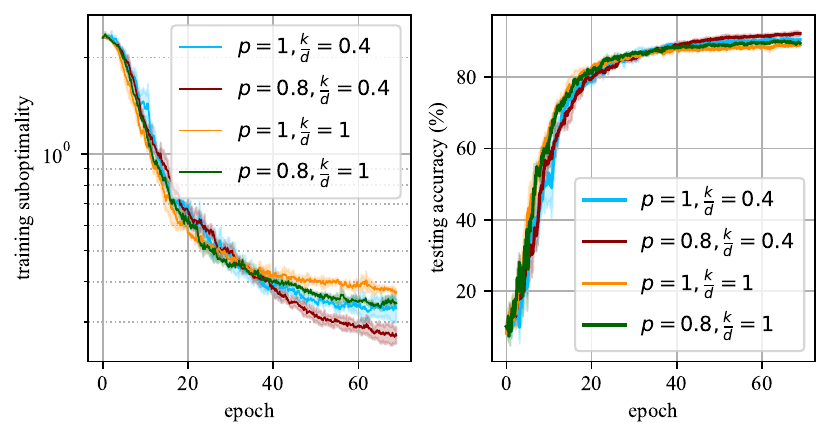}
	\caption{Decentralized CNN training under DO-ADP with $\varepsilon = 1$.
{Communication resource utilization rates per iteration for the yellow, green, blue, and red lines are $100\%$, $80\%$, $40\%$, and $32\%$, respectively.}} 
	\label{fig: cnn}
\end{figure}
\begin{table}[t] 
	\caption{Test accuracy $(\%)$ of decentralized CNN training under different $\varepsilon$.}
	 \centering
  \tiny
	 \begin{tabular}{|m{0.1\linewidth}<{\centering}|m{0.15\linewidth}<{\centering}|m{0.15\linewidth}<{\centering}|m{0.15\linewidth}<{\centering}|m{0.15\linewidth}<{\centering}|}
	 	\hline
	 	\multirow{2}{*}{ }& $p=1$ & $p=0.8$ & $p=1$ & $p=0.8$ \\
    & $\frac{k}{d} = 1$ & $\frac{k}{d} = 1$ & $\frac{k}{d} = 0.4$ & $\frac{k}{d} = 0.4$ \\
	 	\hline
	 	$\varepsilon=1$ & $89.58$   & $90.26$ & $90.98$ & $93.35$ \\
	 	\hline
	 	$\varepsilon=0.9$ & $88.06$   & $90.11$ & $90.75$ & $91.27$ \\
	 	\hline
	 	$\varepsilon=0.8$ & $87.88$   & $87.91$ & $90.11$ & $90.37$ \\
	 	\hline
	 \end{tabular}
  \label{tab: CNN_trade}
\end{table}
Additionally, the testing accuracy tends to increase as communication costs decrease in Fig.~\ref{fig: cnn}, suggesting that efficient communication enhances the trade-off between privacy and accuracy.

\section{Conclusion} \label{sec: conclusions}
This paper investigates the impact of efficient communication on DP for decentralized optimization. 
We propose a novel framework that integrates random agent activation and communication sparsification to simultaneously improve communication efficiency and privacy preservation. 
The DP analysis suggests that increased communication efficiency facilitates stronger privacy guarantees and improved trade-offs among communication, utility, and privacy.

There are several promising directions for future research. 
Firstly, investigating the impact of the stochastic event-triggered mechanism on DP amplification is intriguing, as it introduces additional randomness in communication. 
Additionally, extending the proposed DP amplification framework to directed graphs presents an intriguing opportunity.

\bibliographystyle{IEEEtran}
\bibliography{HW.bib}

\appendices
\section{Some Useful Lemmas}
\label{app: dp_property}

\begin{lemma} \label{lem: subsample}
(Privacy for Subsampling~\cite{steinke2022composition})
	Suppose $\mathcal{M}$ is an $(\varepsilon, \delta)$-DP mechanism. 
	Let $\text{\ttfamily PoisSample}_{p}: \mathcal{U} \to \mathcal{U}$ be the subsampling operation that independently outputs each element in $\mathcal{U}$ with probability $p$, and $\text{\ttfamily UniSample}_{r_{1}, r_{2}}: \mathcal{U}^{r_{1}} \to \mathcal{U}^{r_{2}}$ be the subsampling operation that takes a dataset belonging to $\mathcal{U}^{r_{1}}$ as input and selects uniformly at random a subset of $r_{2} \leq r_{1}$ elements from the input dataset.
	Consider the mechanism
	\begin{equation*}
		\mathcal{A}(\mathcal{D}) = \text{\ttfamily PoisSample}_{p}(\mathcal{R}_{1}, \dots, \mathcal{R}_{n}),
	\end{equation*} 
	where $\mathcal{R}_{i} = \text{\ttfamily DataSample}_{q,1}(\mathcal{M}(\nabla f_{i}(x_{i,t}, \zeta_{i}^{(1)})), \dots, \\ \mathcal{M}(\nabla f_{i}(x_{i,t}, \zeta_{i}^{(q)})))$, is $(\varepsilon^{\prime}, \delta^{\prime})$-DP for $\varepsilon^{\prime} = \log(1+p (e^{\varepsilon} - 1)/q)$ and $\delta^{\prime} = p\delta/q$.
\end{lemma}

\begin{lemma} \label{lem: compose}
	(Composition of DP~\cite{kairouz2015composition})
	Given $T$ randomized algorithms $\mathcal{A}_{1}, \dots, \mathcal{A}_{T}$, each of which is $(\varepsilon_{i}, \delta_{i})$-DP with $\varepsilon_{i} \in (0, 0.9]$ and $\delta_{i} \in (0,1]$. Then $\mathcal{A}$ with $\mathcal{A}(\cdot) = (\mathcal{A}_{1}(\cdot), \dots, \mathcal{A}_{T}(\cdot))$ is $(\tilde{\varepsilon}, \tilde{\delta})$-DP with
	\begin{equation*}
		\tilde{\varepsilon} = \sqrt{\sum_{t=1}^{T}2\varepsilon_{t}^{2} \log \left( e + \frac{\sqrt{\sum_{t=1}^{T} \varepsilon_{t}^{2}}}{\hat{\delta}} \right) } + \sum_{t=1}^{T} \varepsilon_{t}^{2}
	\end{equation*}
	and
	\begin{equation*}
		\tilde{\delta} = 1 - (1-\hat{\delta})\prod_{t=1}^{T}(1-\delta_{t})
	\end{equation*}
	for any $\hat{\delta} \in (0,1]$.
\end{lemma}

\begin{lemma} \label{lem: postp}
	(Post Processing~\cite{dwork2014algorithmic}) 
	Given a randomized algorithm $\mathcal{A}$ that is $(\varepsilon, \delta)$-DP. For arbitrary mapping $h$ from the set of possible outputs of $\mathcal{A}$ to an arbitrary set, $h(\mathcal{A}(\cdot))$ is $(\varepsilon, \delta)$-DP.
\end{lemma}

\begin{lemma} \label{lem: topk}
	For any input vector $x \in \mathbb{R}^{d}$, the output of the $\emph{\text{\ttfamily Top}}_{k}$, $\mathcal{S}(x)$ satisfies
	\begin{equation*}
		\| \mathcal{S}(x) - x \|^{2} \leq \left(1-\frac{k}{d} \right)\|x\|^{2}.
	\end{equation*}
\end{lemma}


\section{Proof of Lemma~\ref{lem: m}}
\label{app: lem_m}

According to~\eqref{eq: m_update}, we have
	\begin{align*}
		m_{i,t} = \sum_{t^{\prime}=0}^{t-1} \beta^{t-1-t^{\prime}}\eta_{i,t^{\prime}}(g_{i,t^{\prime}} + \theta_{i,t^{\prime}}).
	\end{align*}
	Therefore,
	\begin{align*}
		\mathbb{E}\left[ \| m_{i,t} \|^{2} \right] =& \mathbb{E}\left[ \left\| \sum_{t^{\prime}=0}^{t-1} \beta^{t-1-t^{\prime}}\eta_{i,t^{\prime}}(g_{i,t^{\prime}} + \theta_{i,t^{\prime}}) \right \|^{2} \right] \\
		=& \mathbb{E}\left[v_{t}^{2} \left\| \sum_{t^{\prime=0}}^{t} \frac{\beta^{t-1-t^{\prime}}}{v_{t}} \eta_{i,t^{\prime}}(g_{i,t^{\prime}} + \theta_{i,t^{\prime}}) \right\|^{2} \right] \\
		\leq & \frac{1}{1-\beta} \sum_{t^{\prime=0}}^{t} \beta^{t-1-t^{\prime}} \mathbb{E} \left[ \left\| \eta_{i,t^{\prime}}(g_{i,t^{\prime}} + \theta_{i,t^{\prime}})  \right\|^{2} \right] \\
		\leq & \frac{p(G^{2} + \sigma^{2}d)}{(1-\beta)^{2}},
	\end{align*}
	where $v_{t} =\sum_{t^{\prime}=0}^{t-1} \beta^{t-1-t^{\prime}} \leq \frac{1}{1-\beta}$. The first inequality holds from the convexity of $\ell_{2}$ norm.

\section{Proof of Lemma~\ref{lem: x}}
\label{app: lem_x}

We introduce an auxiliary variable, $o_{i,t+1} = x_{i,t} - \eta_{i,t} \alpha m_{i,t+1}$, and thus, $x_{i,t+1} = o_{i,t+1} + \gamma \sum w_{ij} (\hat{x}_{j,t} - \hat{x}_{i,t})$. 
To simplify the presentation, we use $\mathbf{1}$ and $I$ to replace $\mathbf{1}_{n}$ and $I_{n}$.
	Denote 
	\begin{align*}
	X_{t} =& \begin{bmatrix}
		x_{1,t} & x_{2,t} & \cdots & x_{n,t}
	\end{bmatrix} \in \mathbb{R}^{d \times n} \\
	\bar{X}_{t} =& \begin{bmatrix}
		\bar{x}_{t} & \bar{x}_{t} & \cdots & \bar{x}_{t}
	\end{bmatrix} \in \mathbb{R}^{d \times n} \\
	\hat{X}_{t} =& \begin{bmatrix}
			\hat{x}_{1,t} & \dots & \hat{x}_{n,t}
		\end{bmatrix} \in \mathbb{R}^{d \times n} \\
	M_{t} =& \begin{bmatrix}
		m_{1,t} & m_{2,t} & \cdots & m_{n,t}
	\end{bmatrix} \in \mathbb{R}^{d \times n} \\
	O_{t} =& \begin{bmatrix}
			o_{1,t} & \dots & o_{n,t}
		\end{bmatrix} \in \mathbb{R}^{d \times n} \\
		\boldsymbol{\eta}_{t} =& \text{diag}\{\eta_{i,t}, \dots, \eta_{n,t}\} \in \mathbb{R}^{n \times n} \\
		\bar{O}_{t} =& \frac{1}{n}O_{t}\mathbf{1}\mathbf{1}^{\top} \in \mathbb{R}^{n \times n}.
	\end{align*}
According to~\eqref{eq: algo_3}, we have
\begin{subequations}
	\begin{align}
	\label{eq: o_com} O_{t+1} =& X_{t} - \alpha M_{t+1} \boldsymbol{\eta}_{t}, \\
	\label{eq: x_com} X_{t+1} =& O_{t+1} + \gamma \hat{X}_{t}(W-I), \\
	\hat{X}_{t+1} =& \hat{X}_{t} + \mathcal{S}(X_{t+1} - \hat{X}_{t})\boldsymbol{\eta}_{t}, 
\end{align}
\end{subequations}
where
\begin{align*}
	&\mathcal{S}(X_{t+1} - \hat{X}_{t}) \\
		= & \begin{bmatrix}
			\mathcal{S}(x_{i,t+1} - \hat{x}_{i,t}) & \dots & \mathcal{S}(x_{n,t+1} - \hat{x}_{n,t})
		\end{bmatrix} \in \mathbb{R}^{d \times n}.
\end{align*}
Based on~\eqref{eq: o_com} and~\eqref{eq: x_com}, we obtain
\begin{align*}
	\bar{X}_{t+1} = \frac{1}{n}X_{t+1}\mathbf{1}\mathbf{1}^{\top} =  \bar{O}_{t+1} = \bar{X}_{t} - \frac{\alpha}{n}M_{t+1} \boldsymbol{\eta}_{t} \mathbf{1}\mathbf{1}^{\top},
\end{align*}
where the second equality follows from Assumption~\ref{assum: W}. 

There is
\begin{align} \label{eq: x_bar_1}
	& \sum_{i=1}^{n}\mathbb{E}\left[ \|\bar{x}_{t} - x_{i,t} \|^{2} \right] = \mathbb{E}\left[ \left\| X_{t} - \bar{X}_{t} \right\|_{F}^{2} \right] \nonumber \\
	\leq & \mathbb{E}\left[ \left\| X_{t} - \bar{X}_{t} \right\|_{F}^{2} \right] + \mathbb{E}\left[ \left\| X_{t} - \hat{X}_{t} \right\|_{F}^{2} \right].
\end{align}
For the first term of~\eqref{eq: x_bar_1}, we have
\begin{align} \label{eq: x_bar_2}
	& \mathbb{E}\left[ \left\| X_{t} - \bar{X}_{t} \right\|_{F}^{2} \right] \nonumber \\
	=& \mathbb{E}\left[ \left\| O_{t} + \gamma \hat{X}_{t-1}(W-I) - \bar{O}_{t} \right\|_{F}^{2} \right] \nonumber \\
	=& \mathbb{E}\left[ \left\| O_{t} + \gamma \hat{X}_{t-1}(W-I) - \bar{O}_{t} - \gamma \bar{O}_{t}(W-I) \right\|_{F}^{2} \right] \nonumber \\
	=& \mathbb{E}\bigg[ \Big \| (O_{t} - \bar{O}_{t})\left
	[I + \gamma \left(W -I \right) \right] \nonumber \\
	 & \quad \quad  
	 + \gamma (\hat{X}_{t-1} - O_{t})(W-I) \Big \|_{F}^{2} \bigg] \nonumber \\
	 \leq & (1+c_{1}) \|O_{t} - \bar{O}_{t} \|_{F}^{2}  \left\| I + \gamma \left(W -I \right) \right\|_{2}^{2}  \nonumber \\
	 & + \left(1+ \frac{1}{c_{1}}\right) \gamma^{2} \|(\hat{X}_{t-1} - O_{t})(W-I)\|_{F}^{2}, 
\end{align}
where the last inequality is due to $\|a+b\|^{2} \leq (1+\nu)\|a\|^{2} + (1+\nu^{-1})\|b\|^{2}$ for any given vectors $a, b$, and $\nu > 0$. 
For the first term of~\eqref{eq: x_bar_2}, we get
\begin{align*}
	& \|(O_{t} - \bar{O}_{t})\left
	[I + \gamma \left(W -I \right) \right] \|_{F} \\
	\leq & (1-\gamma) \| O_{t} - \bar{O}_{t} \|_{F} + \gamma \| (O_{t} - \bar{O}_{t})W\|_{F} \\
	= & (1-\gamma) \| O_{t} - \bar{O}_{t} \|_{F} + \gamma \left \| (O_{t} - \bar{O}_{t}) \left(W - \frac{1}{n}\mathbf{1}\mathbf{1}^{\top} \right) \right\|_{F} \\
	\leq & (1-\gamma \rho)\| O_{t} - \bar{O}_{t} \|_{F},
\end{align*}
where the last inequality holds from $\|AB \|_{F} \leq \|A\|_{F}\|B\|_{2}$ for given matrix $A, B$, and Assumption~\ref{assum: W}.
Therefore, we obtain
\begin{align} \label{eq: x_bar_4}
	\mathbb{E}\left[ \left\| X_{t} - \bar{X}_{t} \right\|_{F}^{2} \right]
	\leq & (1+c_{1}) (1-\gamma \rho)^{2} \|O_{t} - \bar{O}_{t} \|_{F}^{2} \nonumber \\
	 &+ \left(1+ \frac{1}{c_{1}}\right) \gamma^{2} \phi^{2} \left\| \hat{X}_{t-1} - O_{t}\right \|_{F}^{2}.
\end{align} 


For the second term of~\eqref{eq: x_bar_1}, we have
\begin{align} \label{eq: x_bar_3}
	& \mathbb{E}\left[ \left\| X_{t} - \hat{X}_{t} \right\|_{F}^{2} \right] \nonumber \\
	=& \mathbb{E}\left[ \left\| X_{t} - \hat{X}_{t-1} - \mathcal{S}(X_{t} - \hat{X}_{t-1})\boldsymbol{\eta}_{t-1} \right\|_{F}^{2} \right] \nonumber \\
	= & \sum_{i=1}^{n}\mathbb{E}\left[ \left\| x_{i,t} - \hat{x}_{i,t-1} - \mathcal{S}(x_{i,t} - \hat{x}_{i,t-1}) \eta_{i,t-1} \right \|^{2} \right] \nonumber \\
	\leq & \left(1-\frac{pk}{d} \right) \sum_{i=1}^{n} \| x_{i,t} - \hat{x}_{i,t-1} \|^{2} \nonumber \\
	 =& \left(1-\frac{pk}{d} \right)\| X_{t} - \hat{X}_{t-1} \|_{F}^{2} \nonumber \\
	=& \left(1- \frac{pk}{d} \right) \left\| O_{t} + \gamma \hat{X}_{t-1}(W-I) - \hat{X}_{t-1} \right \|^{2} \nonumber \\
	=& \left(1-\frac{pk}{d} \right) \Big\| (O_{t} - \hat{X}_{t-1})[I - \gamma(W-I)] \nonumber \\
	& \quad \quad \quad \quad \quad -\gamma (O_{t} - \bar{O}_{t})(W-I) \Big \|_{F}^{2} \nonumber \\
	\leq & \left(1-\frac{pk}{d} \right)(1+c_{2})(1+\gamma \phi)^{2}\| O_{t} - \hat{X}_{t-1} \|_{F}^{2} \nonumber \\
	& + \left(1- \frac{pk}{d} \right)\left(1+ \frac{1}{c_{2}} \right)\gamma^{2}\phi^{2}\|O_{t} - \bar{O}_{t} \|_{F}^{2}
\end{align}
Combining~\eqref{eq: x_bar_4} and~\eqref{eq: x_bar_3}, we have
\begin{small}
\begin{align*}
	& \mathbb{E}\left[ \left\| X_{t} - \bar{X}_{t} \right\|_{F}^{2} \right] + \mathbb{E}\left[ \left\| X_{t} - \hat{X}_{t} \right\|_{F}^{2} \right] \\
	\leq & \left[(1+c_{1}) (1-\gamma  \rho)^{2} + \left(1-\frac{pk}{d} \right)\left(1+ \frac{1}{c_{2}} \right)\gamma^{2}\phi^{2} \right] \|O_{t} - \bar{O}_{t} \|_{F}^{2} \\
	 & + \left[ \left(1+ \frac{1}{c_{1}}\right) \gamma^{2} \phi^{2} + \left(1-\frac{pk}{d} \right)(1+c_{2})(1+\gamma \phi)^{2} \right] \| O_{t} - \hat{X}_{t-1} \|_{F}^{2}.
\end{align*}
\end{small}
Similar to~\cite{koloskova2019decentralized}, by choosing $c_{1} = \frac{\gamma \rho}{2}$ and $c_{2} = \frac{p k}{2d}$, and $\gamma = \frac{\rho p k}{d(16\rho + \rho^{2} + 4 \phi^{2} + 2\rho \phi^{2}) - 8\rho p k}$, we can get $c = \frac{\rho^{2}p k}{82d}<1$ such that
\begin{align*}
	& \mathbb{E}\left[ \left\| X_{t} - \bar{X}_{t} \right\|_{F}^{2} \right] + \mathbb{E}\left[ \left\| X_{t} - \hat{X}_{t} \right\|_{F}^{2} \right] \\
	\leq & (1-c) \|O_{t} - \bar{O}_{t} \|_{F}^{2} + (1-c)\| O_{t} - \hat{X}_{t-1} \|_{F}^{2} \\
	=& (1-c) \mathbb{E} \left[\left \|X_{t-1} - \bar{X}_{t-1} - \alpha M_{t}\boldsymbol{\eta}_{t-1}\left(I - \frac{1}{n}\mathbf{1} \mathbf{1}^{\top} \right) \right \|_{F}^{2} \right] \\
	&+ (1-c) \mathbb{E} \left[\left \|X_{t-1} - \alpha M_{t}\boldsymbol{\eta}_{t-1} - \hat{X}_{t-1} \right \|_{F}^{2} \right] \\
	\leq & (1-c)\left(1+\frac{1}{c_{3}} \right) \mathbb{E}\left[ \left \|X_{t-1} - \bar{X}_{t-1} \right \|_{F}^{2} + \left \|X_{t-1} - \hat{X}_{t-1} \right \|_{F}^{2} \right] \\
 &+ (1-c)(1+c_{3})\alpha^{2}\mathbb{E}\left[ 2 \|M_{t} \boldsymbol{\eta}_{t-1} \|_{F}^{2} \right]
\end{align*}
By setting $c_{3} = \frac{2}{c}$ and based on Lemma~\ref{lem: m}, we have
\begin{align*}
	& \mathbb{E}\left[ \left\| X_{t} - \bar{X}_{t} \right\|_{F}^{2} \right] + \mathbb{E}\left[ \left\| X_{t} - \hat{X}_{t} \right\|_{F}^{2} \right] \\
	\leq & \left(1-\frac{c}{2} \right) \mathbb{E}\left[ \left\| X_{t-1} - \bar{X}_{t-1} \right\|_{F}^{2} + \left\| X_{t-1} - \hat{X}_{t-1} \right\|_{F}^{2} \right] \\
	& + \frac{4p\alpha^{2}(G^{2}+ \sigma^{2}d)n}{c(1-\beta)^{2}} \\
	\leq & \frac{8 p \alpha^{2}(G^{2}+ \sigma^{2}d)n}{c^{2}(1-\beta)^{2}},
\end{align*}
where the last step is due to the recursive expansion. Therefore, we have
\begin{align*}
	& \sum_{i=1}^{n}\mathbb{E}\left[ \|\bar{x}_{t} - x_{i,t} \|^{2} \right] = \mathbb{E}\left[ \left\| X_{t} - \bar{X}_{t} \right\|_{F}^{2} \right] \nonumber \\
 \leq &  \frac{8 p \alpha^{2}(G^{2}+ \sigma^{2}d)n}{c^{2}(1-\beta)^{2}}.
\end{align*}


\section{Proof of Theorem~\ref{thm: convergence}}
\label{app: thm_convergence}
According to~\eqref{eq: x_update}, we have
\begin{equation*}
	\bar{x}_{t+1} = \bar{x}_{t} - \frac{\alpha}{n}\sum_{i=1}^{n} \eta_{i,t} m_{i,t+1}.
\end{equation*}
We introduce an auxiliary variable as follows:
\begin{equation*}
	z_{t} = \left\{
	\begin{array}{cl}
		\bar{x}_{t}, & t=0, \\
		\frac{1}{1-\beta}\bar{x}_{t} - \frac{\beta}{1-\beta}\bar{x}_{t-1} , & t \geq 1. \\
	\end{array}
	\right.
\end{equation*}
Then, we can derive the following relationships:
\begin{align} \label{eq: zx}
	z_{t} - \bar{x}_{t}&= \frac{\beta}{1-\beta}(\bar{x}_{t} - \bar{x}_{t-1}) \nonumber \\
	&= -\frac{\alpha\beta}{n(1-\beta)}\sum_{i=1}^{n} \eta_{i,t} m_{i,t+1},
\end{align}
and
\begin{equation} \label{eq: z_update}
\begin{aligned}
	& z_{t+1} - z_{t} \\
	=& \frac{1}{1-\beta}\bar{x}_{t+1} - \frac{1+\beta}{1-\beta}\bar{x}_{t} + \frac{\beta}{1-\beta}\bar{x}_{t-1} \\
	=& -\frac{\alpha}{n(1-\beta)}\sum_{i=1}^{n}[\eta_{i,t}m_{i,t+1} - \beta \eta_{i,t-1}m_{i,t}] \\
	=& -\frac{\alpha}{n(1-\beta)}\sum_{i=1}^{n}\eta_{i,t} (g_{i,t} + \theta_{i,t}) \\
	& - \frac{\alpha\beta}{n(1-\beta)}\sum_{i=1}^{n}(\eta_{i,t}m_{i,t} - \eta_{i,t-1}m_{i,t}).
\end{aligned}  
\end{equation}
We first show some results about $\mathbb{E}\left[ \left\| z_{t} - \bar{x}_{t}\right\|^{2} \right]$ and $\mathbb{E}\left[ \left\| z_{t+1} - z_{t}\right\|^{2} \right]$.

\begin{lemma} \label{lem: zx}
Under Assumption~\ref{assum: f}, for the generated data $\{x_{i,t}, i \in [i] \}$ under Algorithm~\ref{algo: one}, we have that for $t \geq 1$, 
\begin{align*}
	&\mathbb{E}[\|z_{t} - \bar{x}_{t} \|^{2}] \\
\leq & \frac{\alpha^{2}\beta^{2}( \varsigma^{2} + \sigma^{2}d)p}{(1-\beta)^{4}n} \\
&+ \frac{\alpha^{2} \beta^{2}}{(1-\beta)^{3}} \mathbb{E}\left[ \sum_{t^{\prime}=0}^{t} \beta^{t-t^{\prime}}\left\|\frac{1}{n}\sum_{i=1}^{n} \eta_{i,t^{\prime}} \nabla f_{i}(x_{i,t^{\prime}}) \right\|^{2} \right].
\end{align*}
\end{lemma}
\begin{proof}
	The proof is provided in Appendix~\ref{app: lem_zx}.
\end{proof}

\begin{lemma} \label{lem: z_update}
	Under Assumption~\ref{assum: f}, for the generated data $\{x_{i,t}, i \in [i] \}$ under Algorithm~\ref{algo: one}, we have that for $t \geq 1$, there is
	\begin{equation*}
	\begin{aligned}
		&\mathbb{E}\left[ \|z_{t+1} - z_{t} \|^{2} \right]\\
		 \leq & \frac{2p \alpha^{2} (\varsigma^2{ + \sigma^{2}d})}{n(1-\beta)^{2}} + \frac{2\alpha^{2}}{(1-\beta)^{2}} \mathbb{E}\left[ \left\| \frac{1}{n} \sum_{i=1}^{n} \eta_{i,t} \nabla f_{i}(x_{i,t}) \right\|^{2} \right] \\
		 &+ \frac{4\alpha^{2} \beta^{2}(G^{2}+\sigma^{2}d)p^{2}}{n(1-\beta)^{4}}
	\end{aligned}
	\end{equation*}
\end{lemma}
\begin{proof}
	The proof is provided in Appendix~\ref{app: lem_zupdate}.
\end{proof}

According to the $L$-smoothness of $f$, we have
\begin{align} \label{eq: final_1}
	&\mathbb{E}[f(z_{t+1})] \nonumber \\
	\leq & f(z_{t}) + \mathbb{E}\left[ \left< \nabla f(z_{t}), z_{t+1} - z_{t} \right> \right] + \frac{L}{2} \mathbb{E} \left[\|z_{t+1} - z_{t}\|^{2} \right] 
\end{align}
For the second term of~\eqref{eq: final_1}, we have
\begin{align*}
	& \mathbb{E}\left[ \left< \nabla f(z_{t}), z_{t+1} - z_{t} \right> \right] \\
	=& -\frac{\alpha}{n(1-\beta)} \mathbb{E}\left[ \left< \nabla f(z_{t}), \sum_{i=1}^{n}\eta_{i,t} (g_{i,t} + \theta_{i,t}) \right> \right] \\
	& - \frac{\alpha\beta}{n(1-\beta)} \mathbb{E}\left[ \left< \nabla f(z_{t}), \sum_{i=1}^{n}(\eta_{i,t} - \eta_{i,t-1})m_{i,t} \right> \right] \\
	=& -\frac{\alpha}{n(1-\beta)} \mathbb{E}\left[ \left< \nabla f(z_{t}), \sum_{i=1}^{n}\eta_{i,t}  \nabla f_{i}(x_{i,t}) \right> \right] \\
	=& -\frac{\alpha p}{1-\beta} \mathbb{E}\left[ \left< \nabla f(z_{t}) - \nabla f(\bar{x}_{t}) ,\frac{1}{n} \sum_{i=1}^{n}  \nabla f_{i}(x_{i,t}) \right> \right] \\
	& -\frac{\alpha p}{1-\beta} \mathbb{E}\left[ \left< \nabla f(\bar{x}_{t}) ,\frac{1}{n} \sum_{i=1}^{n}  \nabla f_{i}(x_{i,t}) \right> \right] \\
	\leq & \frac{\alpha p}{2\nu_{1}(1-\beta)} \mathbb{E}\left[ \| \nabla f(z_{t}) - \nabla f(\bar{x}_{t})\|^{2} \right] \\
	&+ \frac{\nu_{1}\alpha p}{2(1-\beta)}\mathbb{E}\left[ \left\| \frac{1}{n} \sum_{i=1}^{n}  \nabla f_{i}(x_{i,t}) \right \|^{2} \right]  \\
	&+ \frac{\alpha p}{2(1-\beta)} \mathbb{E}\left[ \left\| \nabla f(\bar{x}_{t}) - \frac{1}{n} \sum_{i=1}^{n}  \nabla f_{i}(x_{i,t}) \right\|^{2} \right] \\
	&- \frac{\alpha p}{2(1-\beta)}\mathbb{E}\left[ \left\| \nabla f(\bar{x}_{t}) \right\|^{2} \right] \\
	& - \frac{\alpha p}{2(1-\beta)}\mathbb{E}\left[ \left\|\frac{1}{n} \sum_{i=1}^{n}  \nabla f_{i}(x_{i,t}) \right\|^{2} \right].
\end{align*}
By setting $\nu_{1} = \frac{\alpha  \beta L}{(1-\beta)^{2}}$ and based on Assumption~\ref{assum: f}, we have
\begin{align} \label{eq: final_2}
	& \mathbb{E}\left[ \left< \nabla f(z_{t}), z_{t+1} - z_{t} \right> \right] \nonumber \\
	\leq & \frac{(1-\beta)Lp}{2\beta } \mathbb{E}\left[ \|z_{t} - \bar{x}_{t} \|^{2} \right] + \frac{\alpha p L^{2}}{2(1-\beta)n} \sum_{i=1}^{n}\mathbb{E}\left[ \|\bar{x}_{t} - x_{i,t} \|^{2} \right] \nonumber \\
	& + \left( \frac{\alpha^{2} p \beta L}{2(1-\beta)^{3}}- \frac{\alpha p}{2(1-\beta)} \right) \mathbb{E}\left[ \left\| \frac{1}{n} \sum_{i=1}^{n}  \nabla f_{i}(x_{i,t}) \right\|^{2} \right] \nonumber \\
	& - \frac{\alpha p}{2(1-\beta)}\mathbb{E}\left[ \left\| \nabla f(\bar{x}_{t}) \right\|^{2} \right] .
\end{align}

Then, we have
\begin{align*}
	& f(z_{t+1}) \\
\leq & f(z_{t}) - \frac{\alpha p}{2(1-\beta)}\mathbb{E}\left[ \left\| \nabla f(\bar{x}_{t}) \right\|^{2} \right] + \frac{(1-\beta)Lp}{2\beta } \mathbb{E}\left[ \|z_{t} - \bar{x}_{t} \|^{2} \right] \\
&+  \frac{\alpha p L^{2}}{2(1-\beta)n} \sum_{i=1}^{n}\mathbb{E}\left[ \|\bar{x}_{t} - x_{i,t} \|^{2} \right] \\
&+ \left( \frac{\alpha^{2} p \beta L}{2(1-\beta)^{3}}- \frac{\alpha p}{2(1-\beta)} \right) \mathbb{E}\left[ \left\| \frac{1}{n} \sum_{i=1}^{n}  \nabla f_{i}(x_{i,t}) \right\|^{2} \right] \\
& + \frac{p \alpha^{2} (\varsigma^2{ + \sigma^{2}d})L}{n(1-\beta)^{2}} + \frac{\alpha^{2}L}{(1-\beta)^{2}} \mathbb{E}\left[ \left\| \frac{1}{n} \sum_{i=1}^{n} \eta_{i,t} \nabla f_{i}(x_{i,t}) \right\|^{2} \right] \\
		 &+ \frac{2\alpha^{2} \beta^{2}(G^{2}+\sigma^{2}d)p^{2}L}{n(1-\beta)^{4}} \\
\leq & f(z_{t}) - \frac{\alpha p}{2(1-\beta)}\mathbb{E}\left[ \left\| \nabla f(\bar{x}_{t}) \right\|^{2} \right] +  \frac{\alpha^{2} \beta (\varsigma^{2}+\sigma^{2}d)Lp}{2(1-\beta)^{3}n} \\
&+ \frac{4\alpha^{3}p^{2}(G^{2} + \sigma^{2}d)L^{2}}{(1-\beta)^{3}c^{2}} \\ 
&+ \left( \frac{\alpha^{2} p \beta L}{2(1-\beta)^{3}}- \frac{\alpha p}{2(1-\beta)} + \frac{\alpha^{2}Lp}{(1-\beta)^{2}} \right) \frac{1}{n}\sum_{i=1}^{n} \mathbb{E}\left[ \left\| \nabla f_{i}(x_{i,t}) \right\|^{2} \right] \\
&+ \frac{\alpha^{2}\beta Lp}{2(1-\beta)^{2}}\mathbb{E}\left[\sum_{t^{\prime}=0}^{t} \beta^{t-t^{\prime}} \left\| \frac{1}{n}\sum_{i=0}^{n} \eta_{i,t^{\prime}} \nabla f_{i}(x_{i,t^{\prime}}) \right\|^{2} \right] \\
& + \frac{p \alpha^{2} (\varsigma^2{ + \sigma^{2}d})L}{n(1-\beta)^{2}}  + \frac{2\alpha^{2} \beta^{2}(G^{2}+\sigma^{2}d)p^{2}L}{n(1-\beta)^{4}},
\end{align*}
Rearranging the above inequality and summing $t$ from $0$ to $T-1$, we can get
\begin{align*}
	& \frac{\alpha p}{2(1-\beta)} \sum_{t=0}^{T-1} \mathbb{E}\left[ \left\| \nabla f(\bar{x}_{t}) \right\|^{2} \right] \\
\leq & f(z_{0}) - \mathbb{E}[f(z_{T})] \\
 &+ \left[ \frac{\alpha^{2} \beta (\varsigma^{2}+\sigma^{2}d)Lp}{2(1-\beta)^{3}n} + \frac{4\alpha^{3}p^{2}(G^{2} + \sigma^{2}d)L^{2}}{(1-\beta)^{3}c^{2}}  \right]T \\
&+ \frac{1}{n} \left( \frac{\alpha^{2} p \beta L}{(1-\beta)^{3}}- \frac{\alpha p}{2(1-\beta)} + \frac{\alpha^{2}Lp}{(1-\beta)^{2}} \right) \sum_{i=1}^{n} \sum_{t=0}^{T-1} \mathbb{E}\left[ \left\| \nabla f_{i}(x_{i,t}) \right\|^{2} \right] \\
& + \left[ \frac{p \alpha^{2} (\varsigma^2{ + \sigma^{2}d})L}{n(1-\beta)^{2}}  + \frac{2\alpha^{2} \beta^{2}(G^{2}+\sigma^{2}d)p^{2}L}{n(1-\beta)^{4}} \right]T,
\end{align*}
where we use $\sum_{t=0}^{T-1}\mathbb{E}\left[\sum_{t^{\prime}=0}^{t} \beta^{t-t^{\prime}} \left\| \frac{1}{n}\sum_{i=0}^{n} \eta_{i,t^{\prime}} \nabla f_{i}(x_{i,t^{\prime}}) \right\|^{2} \right] \\ \leq \frac{1}{n(1-\beta)}\sum_{t=0}^{T-1}\sum_{i=0}^{n} \mathbb{E}\left[ \left\| \nabla f_{i}(x_{i,t}) \right\|^{2} \right]$.
By setting $\alpha < \frac{(1-\beta)^{2}}{2L}$, we can get
\begin{align*}
    & \frac{1}{T} \sum_{t=0}^{T-1} \mathbb{E}\left[ \left\| \nabla f(\bar{x}_{t}) \right\|^{2} \right] \\
\leq & \frac{2(1-\beta)(f(x_{0}) - f^{*})}{\alpha p T} \\
&+ \frac{\alpha L(\beta + 2 + 4 \beta^{2}p)(G^{2} + \varsigma^{2} + \sigma^{2}d)}{n(1-\beta)^{3}} \\
& + \frac{8\alpha^{2}p(G^{2} + \sigma^{2}d)L^{2}}{(1-\beta)^{2}c^{2}}
\end{align*}



\section{Proof of Lemma~\ref{lem: zx}}
\label{app: lem_zx}

Based on~\eqref{eq: zx}, we obtain
\begin{small}
	\begin{align*}
		&\mathbb{E}[\|z_{t} - \bar{x}_{t} \|^{2}] = \frac{\alpha^{2} \beta^{2}}{(1-\beta)^{2}}\mathbb{E}\left[ \left\| \frac{1}{n} \sum_{i=1}^{n} \eta_{i,t} m_{i,t+1} \right\|^{2} \right]
	\end{align*}
According to the definition of $m_{i,t}$, we have
\begin{align*}
	& \mathbb{E}\left[ \left\| \frac{1}{n} \sum_{i=1}^{n} \eta_{i,t} m_{i,t+1} \right\|^{2} \right] \\
	 \leq & \mathbb{E} \left[ \left\| \sum_{t^{\prime}=0}^{t} \beta^{t-t^{\prime}} \frac{1}{n} \sum_{i=1}^{n} \eta_{i,t^{\prime}} (g_{i,t^{\prime}} + \theta_{i,t^{\prime}}) \right\|^{2} \right] \\
	 \leq & v_{t} \sum_{t^{\prime}=0}^{t} \beta^{t-t^{\prime}} \mathbb{E}\left[ \left\| \frac{1}{n} \sum_{i=1}^{n} \eta_{i,t^{\prime }}(g_{i,t^{\prime}} + \theta_{i,t^{\prime}}) \right\|^{2}  \right] \\
	 \leq & \frac{1}{1-\beta} \sum_{t^{\prime}=0}^{t} \beta^{t-t^{\prime}}  \left( \frac{p(\varsigma^{2} + \sigma^{2}d)}{n} + \mathbb{E}\left[\left\| \frac{1}{n} \sum_{i=1}^{n} \eta_{i,t^{\prime}} \nabla f_{i}(x_{i,t^{\prime}}) \right\|^{2} \right] \right)  \\
	 \leq & \frac{p(\varsigma^{2} + \sigma^{2}d)}{n(1-\beta)^{2}} + \frac{1}{1-\beta}\mathbb{E}\left[\sum_{t^{\prime}=0}^{t} \beta^{t-t^{\prime}} \left\| \frac{1}{n} \sum_{i=1}^{n} \eta_{i,t^{\prime}} \nabla f_{i}(x_{i,t^{\prime}}) \right\|^{2}  \right].
\end{align*}
\end{small}
Therefore, we get the result shown in Lemma~\ref{lem: zx}.

\section{Proof of Lemma~\ref{lem: z_update}}
\label{app: lem_zupdate}

Based on the update of $z_{t}$ in~\eqref{eq: z_update}, we have
\begin{align} \label{eq: z_1}
	& \mathbb{E}\left[\|z_{t+1} - z_{t} \|^{2}\right] \nonumber \\
	\leq & \frac{2\alpha^{2}}{(1-\beta)^{2}}\mathbb{E}\left[ \left\| \frac{1}{n} \sum_{i=1}^{n} \eta_{i,t}(g_{i,t} + \theta_{i,t}) \right\|^{2} \right] \nonumber \\
	& + \frac{2\alpha^{2}\beta^{
    2}}{n^{2}(1-\beta)^{2}}\mathbb{E}\left[ \left\| \sum_{i=1}^{n} m_{i,t}(\eta_{i,t} - \eta_{i,t-1}) \right\|^{2} \right] 
\end{align}
For the first term of~\eqref{eq: z_1}, we have
\begin{align*}
	& \mathbb{E}\left[ \left\| \frac{1}{n} \sum_{i=1}^{n} \eta_{i,t}(g_{i,t} + \theta_{i,t}) \right\|^{2} \right] \\
	=& \mathbb{E}\left[ \left\| \frac{1}{n} \sum_{i=1}^{n} \eta_{i,t}(g_{i,t} - \nabla f_{i}(x_{i,t}) ) \right\|^{2}\right] \\
	&+ \mathbb{E}\left[ \left\| \frac{1}{n} \sum_{i=1}^{n} \eta_{i,t} \nabla f_{i}(x_{i,t}) \right\|^{2} \right] + \mathbb{E}\left[ \left\| \frac{1}{n} \sum_{i=1}^{n} \eta_{i,t} \theta_{i,t} \right\|^{2} \right] \\
	\leq & \frac{p(\varsigma^2{ + \sigma^{2}d})}{n} + \mathbb{E}\left[ \left\| \frac{1}{n} \sum_{i=1}^{n} \eta_{i,t} \nabla f_{i}(x_{i,t}) \right\|^{2} \right].
\end{align*}
For the second term of~\eqref{eq: z_1}, we have
\begin{equation*}
	(\eta_{i,t} - \eta_{i,t-1})m_{i,t} = \left\{
	\begin{array}{cl}
		0, & w.p. \ p^{2}+(1-p)^{2}, \\
		m_{i,t}, & w.p.\ p(1-p), \\
		-m_{i,t}, & w.p. \ p(1-p),
	\end{array}
	\right.
\end{equation*}
and thus, $\mathbb{E}[(\eta_{i,t} - \eta_{i,t-1})m_{i,t}]=0$, and $\mathbb{E}[\| (\eta_{i,t} - \eta_{i,t-1})m_{i,t} \|^{2}]\leq \frac{2p(1-p)(G+\sigma^{2}d)}{(1-\beta)^{2}} \leq \frac{2p^{2}(G+\sigma^{2}d)}{(1-\beta)^{2}}$ based on Lemma~\ref{lem: m}. Therefore, we get the result shown in Lemma~\ref{lem: z_update}.

\section{Proof of Corollary~\ref{cor: optimal}}
\label{app: cor_optimal} 
To obtain the final convergence rate, we carefully tune the step size. First, we consider an auxiliary lemma.
\begin{lemma}\cite{koloskova2019decentralized} \label{lem: step_tune}
        For any parameters $r_{0} \geq 0$, $b \geq 0$, $h \geq 0$, $d \geq 0$, there exists a constant step size $\alpha \leq \frac{1}{d}$ such that
    \begin{align*}
        \Phi_{T} :=& \frac{r_{0}}{\alpha (T+1)} + b \alpha + h \alpha^{2} \\
        \leq & 2 \left( \frac{br_{0}}{T+1}\right)^{\frac{1}{2}} + 2h^{\frac{1}{3}} \left( \frac{r_{0}}{T+1} \right)^{\frac{2}{3}} + \frac{dr_{0}}{T+1}.
    \end{align*}
\end{lemma}

Based on Theorem~\ref{thm: convergence} and Lemma~\ref{lem: step_tune}, 
with $r_{0} = \frac{2(1-\beta)(f(x_{0}) - f^{*})}{p}$, $b = \frac{L(\beta + 2 + 4 \beta^{2}p)(G^{2} + \varsigma^{2} + \sigma^{2}d)}{n(1-\beta)^{3}}$, and $h = \frac{8p(G^{2} + \sigma^{2}d)L^{2}}{(1-\beta)^{2}c^{2}}$,
choosing $\alpha = \min  \left\{ \left(\frac{r_{0}}{b(T+1)} \right)^{\frac{1}{2}}, \left( \frac{r_{0}}{h(T+1)}\right)^{\frac{1}{3}}, \frac{(1-\beta)^{2}}{2L} \right \} \leq \frac{(1-\beta)^{2}}{2L}$,
we have
\begin{align*}
     & \frac{1}{T} \sum_{t=0}^{T-1} \mathbb{E}\left[ \left\| \nabla f(\bar{x}_{t}) \right\|^{2} \right] \\
\leq & 2 \left( \frac{2L(\beta + 2 + 4 \beta^{2}p)(G^{2} + \varsigma^{2} + \sigma^{2}d) (f(x_{0}) - f^{*})   }{n(1-\beta)^{2}pT}  \right)^{\frac{1}{2}} \\
& + 2 \left( \frac{8p(G^{2} + \sigma^{2}d)L^{2}}{c^{2}}  \right)^{\frac{1}{3}} \left( \frac{2(f(x_{0}) - f^{*})}{pT} \right)^{\frac{2}{3}} \\
&+ \frac{4L(f(x_{0}) - f^{*})}{(1-\beta)pT} \\
\leq & O \left(   \left( \frac{(1+p)(1+\sigma^{2}d)}{npT}  \right)^{\frac{1}{2}} +  
\left(   \frac{\sqrt{p(1+\sigma^{2}d)}}{cpT}  \right)^{\frac{2}{3}} + \frac{1}{pT}
\right).
\end{align*}

With constant step size $\alpha = \sqrt{\frac{n}{T}}$ for $T \geq \frac{4nL^{2}}{(1-\beta)^{4}}$, there is
\begin{align*}
    & \frac{1}{T} \sum_{t=0}^{T-1} \mathbb{E}\left[ \left\| \nabla f(\bar{x}_{t}) \right\|^{2} \right] \\
\leq & \frac{2(f(x_{0}) - f^{*})(1-\beta)}{p\sqrt{nT}} + \frac{ L(\beta + 2 + 4 \beta^{2}p)(G^{2} + \varsigma^{2} + \sigma^{2}d) }{(1-\beta)^{3}\sqrt{nT}} \\
&+ \frac{8np(G^{2} + \sigma^{2}d)L^{2}}{(1-\beta)^{2}c^{2}T}.
\end{align*}

\end{document}